\documentclass[preprint,3p,number,sort&compress]{elsarticle}

\usepackage{amssymb}
\usepackage{amsmath}
\usepackage{amsthm}
\usepackage[pdfencoding=auto, psdextra,hidelinks]{hyperref}
\usepackage{lipsum}
\usepackage{amsfonts}
\usepackage{algorithm}
\usepackage{algorithmic}
\usepackage{mathtools}
\usepackage{cleveref}
\usepackage[dvipsnames,table]{xcolor}

\newtheorem{theorem}{Theorem}[section]
\newtheorem{lemma}[theorem]{Lemma}

\newtheorem{proposition}[theorem]{Proposition}
\newtheorem{definition}[theorem]{Definition}
\newtheorem{remark}[theorem]{Remark}
\ifpdf
 \DeclareGraphicsExtensions{.eps,.pdf,.png,.jpg}
\else
 \DeclareGraphicsExtensions{.eps}
\fi

\newcommand{\bd}{\textbf{d}}
\newcommand{\be}{\textbf{e}}

\newcommand{\bu}{\textbf{u}}

\newcommand{\bx}{\textbf{x}}

\newcommand{\bA}{\textbf{A}}
\newcommand{\bB}{\textbf{B}}

\newcommand{\bD}{\textbf{D}}
\newcommand{\bE}{\textbf{E}}
\newcommand{\bF}{\textbf{F}}

\newcommand{\bH}{\textbf{H}}
\newcommand{\bI}{\textbf{I}}

\newcommand{\bX}{\textbf{X}}
\newcommand{\bY}{\textbf{Y}}
\newcommand{\bZ}{\textbf{Z}}

\newcommand{\bQ}{\textbf{Q}}
\newcommand{\bR}{\textbf{R}}

\newcommand{\bU}{\textbf{U}}

\newcommand{\bDelta}{\boldsymbol{\Delta}}

\newcommand{\bLam}{\boldsymbol{\Lambda}}

\newcommand{\cblue }{\color{black}}
\newcommand{\cn}{\color{black}}
\newcommand{\apx}[1]{\underline {#1}}

\DeclarePairedDelimiter\norm{\lVert}{\rVert\,}%
\DeclarePairedDelimiter\abs{\lvert}{\rvert\,}%

\journal{Computer Physics Communications}

\begin{document}

\begin{frontmatter}

%% Title, authors and addresses

%% use the tnoteref command within \title for footnotes;
%% use the tnotetext command for theassociated footnote;
%% use the fnref command within \author or \affiliation for footnotes;
%% use the fntext command for theassociated footnote;
%% use the corref command within \author for corresponding author footnotes;
%% use the cortext command for theassociated footnote;
%% use the ead command for the email address,
%% and the form \ead[url] for the home page:
%% \title{Title\tnoteref{label1}}
%% \tnotetext[label1]{}
%% \author{Name\corref{cor1}\fnref{label2}}
%% \ead{email address}
%% \ead[url]{home page}
%% \fntext[label2]{}
%% \cortext[cor1]{}
%% \affiliation{organization={},
%%       addressline={},
%%       city={},
%%       postcode={},
%%       state={},
%%       country={}}
%% \fntext[label3]{}

\title{Residual-based Chebyshev filtered subspace iteration for Hermitian eigenvalue problems tolerant to inexact matrix-vector products} 

%% use optional labels to link authors explicitly to addresses:
%% \author[label1,label2]{}
%% \affiliation[label1]{organization={},
%%       addressline={},
%%       city={},
%%       postcode={},
%%       state={},
%%       country={}}
%%
%% \affiliation[label2]{organization={},
%%       addressline={},
%%       city={},
%%       postcode={},
%%       state={},
%%       country={}}

\author[inst1]{Nikhil Kodali}
\author[inst1]{Kartick Ramakrishnan}
\author[inst1]{Phani Motamarri}
\ead{phanim@iisc.ac.in}
\affiliation[inst1]{organization={Department of Computational and Data Sciences, Indian Institute of Science},%Department and Organization
      addressline={CV Raman Road}, 
      city={Bengaluru},
      postcode={560012}, 
      state={Karnataka},
      country={India}}

%% Abstract
\begin{abstract}

{\cblue Chebyshev Filtered Subspace Iteration (ChFSI) is widely used for computing a small subset of extremal eigenpairs from large matrices, particularly when the eigenpairs must be computed repeatedly as the system matrix evolves within an outer nonlinear iteration. These evolving subspaces are encountered, for instance, in Kohn-Sham density functional theory (DFT) used in quantum modelling of materials, subspace tracking in signal processing applications, and matrix completion problems. In this work, we propose R-ChFSI, a residual-based reformulation that recasts the Chebyshev polynomial recurrence in terms of residuals rather than eigenvector estimates, that achieves robust convergence even when matrix--vector products are computed inexactly. We derive convergence guarantees under such approximations and show that R-ChFSI can naturally leverage (i) the use of inexpensive approximate inverses for generalized eigenproblems of the form $\bA \bx = \lambda \bB \bx$, where exact factorizations of $\bB$ are prohibitively expensive, (ii)~low-precision arithmetic (FP32, TF32) for both standard and generalized eigenproblems, and (iii)~reduced-precision (BF16) inter-process communication in distributed sparse matrix--vector products. Controlled experiments on dense random matrices quantitatively verify the convergence bounds derived in this work and confirm the robustness of R-ChFSI to prescribed approximation errors for both standard and generalized eigenproblems. Large-scale experiments on finite-element discretized DFT generalized eigenproblems with up to 85 million grid points and 13,500 eigenpairs demonstrate that R-ChFSI achieves residual norms orders of magnitude below those of standard ChFSI when approximate inverses are employed, and reliably meets target tolerances of $10^{-8}$ even when employing reduced precision, yielding filtering speedups of up to $2.7{\times}$ ($2.1{\times}$ for the full eigensolver) on GPU accelerators.\cn}

\end{abstract}

% %%Graphical abstract
% \begin{graphicalabstract}
% %\includegraphics{grabs}
% \end{graphicalabstract}

% %%Research highlights
% \begin{highlights}
% \item Research highlight 1
% \item Research highlight 2
% \end{highlights}

%% Keywords
\begin{keyword}
%% keywords here, in the form: keyword \sep keyword

%% PACS codes here, in the form: \PACS code \sep code

%% MSC codes here, in the form: \MSC code \sep code
%% or \MSC[2008] code \sep code (2000 is the default)
Hermitian generalized eigenproblems, Chebyshev filtered subspace iteration, finite-element basis, Kohn-Sham density functional theory 
\end{keyword}

\end{frontmatter}

%% Add \usepackage{lineno} before \begin{document} and uncomment 
%% following line to enable line numbers
%% \linenumbers

%% main text
%%

%% Use \section commands to start a section
\section{Introduction}
Large Hermitian eigenvalue problems play a central role in many areas of computational physics. They arise from the discretization of fundamental governing equations in quantum mechanics, electronic structure theory, fluid dynamics, plasma physics, wave propagation, electromagnetics, and elastodynamics. In such applications, the extremal eigenmodes often correspond to physically meaningful states—ground and low-lying excited states in quantum systems, dominant instability modes in fluid or plasma flows, or fundamental vibration frequencies in elastic media—which makes the efficient computation of only a small subset of the spectrum particularly important. The solution strategies for these eigenproblems often rely on fully iterative methods based on iterative orthogonal projection approaches. In these approaches, the large matrix is orthogonally projected onto a carefully constructed smaller subspace rich in the wanted eigenvectors (Rayleigh-Ritz step), followed by subspace diagonalisation of the projected matrix and a subspace rotation step to recover the desired orthogonal eigenvector estimates of the original Hermitian matrix. Popular iterative approaches include Davidson~\cite{crouzeixDavidsonMethod1994,davidsonIterativeCalculationFew1975}, Generalized-Davidson~\cite{morganGeneralizationsDavidsonsMethod1986}, Jacobi-Davidson~\cite{hochstenbachJacobiDavidsonMethod2006}, Chebyshev-filtered subspace iteration (ChFSI) approach~\cite{zhouSelfconsistentfieldCalculationsUsing2006}, LOBPCG~\cite{knyazevOptimalPreconditionedEigensolver2001}, and PPCG~\cite{vecharynskiProjectedPreconditionedConjugate2015}. Another key class of iterative techniques is based on Krylov subspace methods, namely the Arnoldi method~\cite{arnoldiPrincipleMinimizedIterations1951}, Lanczos methods~\cite{lanczosIterationMethodSolution1950} and their important variants, including implicit restart Arnoldi methods~\cite{sorensenImplicitApplicationPolynomial1992}, Krylov-Schur method~\cite{stewartKrylovSchurAlgorithmLarge2002} and block-Krylov methods~\cite{saadNumericalMethodsLarge2011}. Our focus in this work involves the solution of large Hermitian eigenproblems using the ChFSI approach.

Chebyshev filtered subspace iteration (ChFSI) has emerged as a robust alternative to Krylov subspace eigensolvers for extracting a small subset of extremal eigenpairs from large Hermitian matrices—particularly in scenarios where these eigenpairs must be computed repeatedly as the system matrix evolves within an outer iteration. A prominent example is the nonlinear eigenvalue problem that arises in Kohn-Sham density functional theory (DFT)~\cite{kohnSelfconsistentEquationsIncluding1965}, where one is interested in the lowest eigenpairs of a Hamiltonian that depends nonlinearly on the eigenvectors. This problem is typically solved through a Self-Consistent Field (SCF) procedure, in which the Hamiltonian is updated at each iteration and a small set of extremal eigenpairs must be recomputed repeatedly as the system approaches self-consistency. ChFSI relies on a Chebyshev polynomial filtering procedure that constructs a subspace rich in the desired eigenvectors, exploiting the fast growth property of Chebyshev polynomials outside the interval [-1,1], followed by a Rayleigh-Ritz step of projection and subspace diagonalization. Approaches based on ChFSI have become quite popular recently to solve both standard and generalized eigenproblems arising in DFT~\cite{lohChebyshevFilteringLanczos1984,zhouParallelSelfconsistentfieldCalculations2006,motamarriHigherorderAdaptiveFiniteelement2013,levittParallelEigensolversPlanewave2015,wuAdvancingDistributedMultiGPU2023}. Although variants of ChFSI approaches combined with the Davidson method~\cite{zhouChebyshevDavidsonAlgorithm2007,miaoChebyshevDavidsonMethod2020,miaoFlexibleBlockChebyshevDavidson2023,wangNewSubspaceIteration2025} have been proposed, the ChFSI approach has become the preferred choice in large-scale electronic structure codes for solving the underlying nonlinear eigenvalue problem~\cite{dasDFTFE10Massively2022,ghoshSPARCAccurateEfficient2017a,zhouChebyshevfilteredSubspaceIteration2014,fengMassivelyParallelImplementation2024,banerjeeChebyshevPolynomialFiltered2016} due to its scalability, ability to accommodate evolving subspaces and computational efficiency~\cite{dasFastScalableAccurate2019,dasLargeScaleMaterialsModeling2023}. {\cblue We note that planewave-based DFT codes often employ preconditioned iterative solvers (e.g., preconditioned conjugate gradient or Davidson-type methods) that exploit the availability of effective kinetic-energy-based preconditioners in the Fourier basis. However, for real-space discretizations (finite-element, finite-difference) where constructing such spectrally adapted preconditioners is less straightforward, ChFSI has proven to be the most effective approach for large-scale calculations. Additionally, ChFSI is memory efficient compared to Davidson-type approaches and further offers an attractive alternative to preconditioned-type conjugate gradient approaches for problems where good preconditioners are unavailable.\cn} The current work fills an important gap in the ChFSI approach by developing a residual-based approach to ChFSI that is tolerant to inexact matrix-vector products during the subspace construction step. Consequently, we argue that the method is well-suited for solving generalized eigenproblems efficiently and additionally low-precision arithmetic can be leveraged in the light of recent changes in modern heterogeneous computing architectures, significantly improving computational efficiency.

%Most importantly, for generalized eigenvalue problems of the form $\bA\bx=\lambda\bB\bx$, we demonstrate that the method naturally admits low-cost approximations of $\bB^{-1}$, allowing ChFSI to act on $\bB^{-1}\bA$ and solve large sparse generalized eigenproblems without degrading convergence, thereby avoiding Cholesky factorization of $\bB$ that is often impractical due to excessive fill-in. Additionally, we argue that such an approach can be robust while leveraging low-precision matrix-vector products in the subspace construction step. These capabilities are especially important in light of recent changes in modern .

Modern hardware architectures have undergone substantial modifications in recent years due to the computational requirements of machine learning (ML) and artificial intelligence (AI) training. Owing to their high computational demands, these domains have gravitated towards the use of low-precision arithmetic for training and inferencing. In response to this demand, hardware manufacturers have been enhancing support for low-precision floating-point formats, such as tensorfloat32 and bfloat16, enabling significantly faster throughput and substantial performance improvements\footnote{NVIDIA's Blackwell GPUs, designed for AI/ML applications, demonstrate a notable decrease in peak double-precision (FP64) floating-point performance~\cite{HGXAISupercomputing} compared to the previous Hopper architecture.}. These architectural changes underscore the necessity to modify scientific computing algorithms to efficiently utilize low-precision processes without compromising accuracy~\cite{dongarraHardwareTrendsImpacting2024,kashiMixedprecisionNumericsScientific2025}. For sparse eigensolvers, the ideas of mixed precision preconditioning for the LOBPCG eigensolver~\cite{kressnerMixedPrecisionLOBPCG2023} and mixed precision orthogonalization and Rayleigh-Ritz have been explored for the LOBPCG and ChFSI eigensolvers~\cite{kressnerMixedPrecisionLOBPCG2023,dasDFTFE10Massively2022,motamarriDFTFEMassivelyParallel2020}. We note that in these works the evaluation of matrix-vector products involved in the subspace construction still need to be performed in the higher precision arithmetic. The development of robust iterative eigensolver algorithms to accommodate low precision arithmetic in matrix-vector multiplications becomes particularly important in light of the evolving hardware landscape.
%We further note that the capacity to accommodate inexact matrix-vector products facilitated by low-precision arithmetic can yield substantial performance enhancements. Algorithms that can accommodate inexact matrix-vector products while ensuring that errors remain within acceptable tolerances can efficiently utilize the capabilities of modern hardware, thereby reducing computational costs. The development of such algorithms becomes particularly important in light of the evolving hardware landscape.

In this work, we introduce a residual-based reformulation of the Chebyshev filtered subspace iteration (ChFSI) method~\cite{saadChebyshevAccelerationTechniques1984,zhouParallelSelfconsistentfieldCalculations2006,zhouSelfconsistentfieldCalculationsUsing2006,zhouChebyshevfilteredSubspaceIteration2014}, referred to as the R-ChFSI method, for solving large-scale Hermitian eigenvalue problems. The key novelty of R-ChFSI lies in its ability to accommodate inexact matrix-vector products in the subspace construction step while preserving convergence properties. Most importantly, for generalized eigenvalue problems of the form $\bA\bx=\lambda\bB\bx$, we demonstrate that the method naturally admits low-cost approximations of $\bB^{-1}$, allowing ChFSI to act on $\bB^{-1}\bA$ and solve large generalized eigenproblems without degrading convergence, thereby avoiding 
expensive matrix factorizations and iterative solvers that are otherwise typically required for such problems~\cite{levittParallelEigensolversPlanewave2015}. The ability of R-ChFSI to be robust to approximations in matrix-vector multiplications makes the proposed method particularly relevant for modern hardware architectures that favor low-precision arithmetic, significantly improving computational efficiency. The core of our approach is a reformulated recurrence relation that modifies the standard ChFSI update step to operate on residuals rather than the guess of the eigenvectors. We then provide a mathematical justification demonstrating that this reformulation reduces the numerical error in the Chebyshev filtered subspace construction while employing inexact matrix-vector products compared to the traditional ChFSI recurrence relation. Our analysis establishes that the proposed R-ChFSI method effectively controls error propagation in the Chebyshev filtering process, resulting in more reliable convergence.

Through extensive benchmarking on large-scale sparse generalized eigenvalue problems arising in \emph{ab initio} calculations using finite-element (FE) discretization of DFT, where we employ a diagonal approximation to the FE basis overlap matrix inverse, R-ChFSI attains a residual tolerance that is orders of magnitude lower than standard ChFSI, demonstrating superior robustness to approximation errors. The improved convergence for generalized eigenvalue problems is particularly noteworthy, as it highlights the ability of R-ChFSI to accommodate approximate inverses while maintaining accuracy. In contrast, standard ChFSI is more sensitive to such approximations, often leading to a loss in accuracy or requiring significantly higher computational costs to achieve similar residual tolerances. This advantage of R-ChFSI is especially relevant in DFT calculations, where generalized eigenvalue problems frequently arise~\cite{levittParallelEigensolversPlanewave2015,ramakrishnanFastScalableFiniteelement2025,dasDFTFE10Massively2022}, and high computational costs in the subspace construction constitute a significant bottleneck. Furthermore, R-ChFSI enables the use of lower-precision arithmetic in the filtering step, further enhancing the performance gains. In our benchmark studies, using TF32 arithmetic in Intel Data Center GPU Max Series accelerators deployed on the Aurora supercomputing system, we have obtained speedups of up to 2.3x for the filtering step and 1.9x for the complete eigensolve to reach the desired tolerance. Further, we have also employed reduced precision, BF16, for the nearest-neighbour MPI communication arising in sparse matrix--vector multiplications to enhance performance, and we obtain speedups of up to 2.7x for the filtering step and 2.1x for the complete eigensolve to reach the desired tolerance. By reducing dependence on exact matrix factorizations and high-precision arithmetic, R-ChFSI enables efficient exploitation of modern hardware accelerators, making large-scale eigenvalue computations more feasible in high-performance computing environments.

The remainder of this article is organized as follows: Section 2 outlines the key steps of the Chebyshev filtered subspace iteration procedure and subsequently analyzes the convergence properties of ChFSI in terms of how the maximum principal angle between the current subspace and the target eigenspace evolves during the iterations. Section 3 begins by analyzing the convergence of ChFSI 
when the subspace is constructed approximately due to inexact matrix-vector products. Subsequently, this section introduces the proposed residual-based Chebyshev filtered subspace iteration method (R-ChFSI) for generalized eigenvalue problems. This section discusses the convergence of R-ChFSI with approximate matrix products and demonstrates mathematically that the R-ChFSI method can converge even when the traditional ChFSI method fails. Section 4 presents a comprehensive evaluation of the proposed R-ChFSI method in terms of accuracy and computational efficiency on GPU architectures that admit low precision arithmetic. It compares R-ChFSI to the traditional ChFSI approach for solving real symmetric and complex Hermitian eigenproblems in the case of generalized eigenproblems while using inexact matrix-vector products.

\section{Mathematical Background}
\label{sec:background}
The ChFSI \cite{saadChebyshevAccelerationTechniques1984,saadNumericalMethodsLarge2011,zhouParallelSelfconsistentfieldCalculations2006,zhouSelfconsistentfieldCalculationsUsing2006,motamarriHigherorderAdaptiveFiniteelement2013,zhouChebyshevfilteredSubspaceIteration2014} approach for solving the desired eigenpairs belongs to the category of iterative orthogonal projection methods and is one of the widely used strategies to compute the smallest $n$ eigenvalues and their corresponding eigenvectors~\cite{dektorInexactSubspaceProjection2025,dasDFTFE10Massively2022,motamarriDFTFEMassivelyParallel2020,zhangSPARCV200Spinorbit2024,ghoshSPARCAccurateEfficient2017a,xuSPARCSimulationPackage2021,kronikPARSECPseudopotentialAlgorithm2006,liouScalableImplementationPolynomial2020,winkelmannChASE2019,michaud-riouxRESCURealSpace2016,doganSolvingElectronicStructure2023,levittParallelEigensolversPlanewave2015}. To describe the proposed eigensolver strategy based on the ChFSI approach, in the current work, we consider the Hermitian generalized eigenvalue problem of the form 
\begin{equation}
\bA\bu_i=\lambda_i\bB\bu_i \;\; \text{where} \;\; \bA\in \mathbb{C}^{m\times m},\; \bB\in \mathbb{C}^{m\times m}.\label{eqn:ghep}
\end{equation}
where $\bA$ and $\bB$ are Hermitian matrices with $\bB$ being a positive-definite matrix. In addition, $\lambda_i\in \mathbb{R}$ and $\bu_i\in\mathbb{C}^{m}:\forall i=1,\dots,n$ denote the eigenvalue-eigenvector pairs corresponding to the smallest $n$ eigenvalues. \Cref{eqn:ghep} reduces to a standard eigenvalue problem when $\bB=\bI$ with $\bI$ denoting the $m\times m$ identity matrix. Without loss of generality, we assume that the eigenvalues are ordered as $\lambda_1\leq\lambda_2\leq\dots\leq\lambda_n\leq\lambda_{n+1}\leq\dots\leq\lambda_m$. For completeness and to introduce notations, we now provide a brief overview of the ChFSI algorithm, traditionally used to solve the eigenvalue problem corresponding to \cref{eqn:ghep}.

\subsection{Chebyshev filtered subspace iteration and convergence properties}

ChFSI leverages the properties of Chebyshev polynomials to efficiently filter out the components of the unwanted eigenvectors (corresponding to the remaining $m-n$ largest eigenvalues), thus enriching the trial subspace with the desired eigenvectors. To this end, we define the Chebyshev polynomial of degree $k$ as $T_k(x)$ and note that these polynomials exhibit the fastest growth in their magnitude when $\abs{x}>1$ while remaining bounded between $[-1,1]$ when $\abs{x}\leq 1$. To exploit this fast growth property of $T_k(x)$, an affine transformation that maps the largest $m-n$ eigenvalues to $[-1,1]$ is defined, and consequently, the desired smallest $n$ eigenvalues get mapped to values lying in $(-\infty,-1]$. To this end, we define the center of the unwanted spectrum as $c=(\lambda_{m}+\lambda_{n+1})/2$ and the half-width of the unwanted spectrum as $e=(\lambda_{m}-\lambda_{n+1})/2$ and hence the required affine transformation can now be represented as $\mathcal{L}(x)= (x-c)/e$. Standard implementations of ChFSI also scale the Chebyshev polynomials to prevent overflow~\cite{saadChebyshevAccelerationTechniques1984,saadNumericalMethodsLarge2011,zhouParallelSelfconsistentfieldCalculations2006,zhouSelfconsistentfieldCalculationsUsing2006,zhouChebyshevfilteredSubspaceIteration2014,winkelmannChASE2019} and consequently the scaled and shifted Chebyshev polynomials are defined as $C_k(x)= T_k(\mathcal{L}(x))/T_k(\mathcal{L}(a_{low}))$ where $a_{low}\leq\lambda_1$. In practice, the exact eigenvalues are unknown, and the filter is parameterised by estimates: a lower bound $\lambda_{\min}\leq\lambda_1$ (taking the role of~$a_{low}$), an upper bound $\lambda_{\max}\geq\lambda_m$ on the spectrum, and a threshold $\lambda_T$ separating the wanted from the unwanted eigenvalues (so that $\lambda_T\approx\lambda_{n+1}$). In terms of these estimates the affine-transformation parameters become $c=(\lambda_{\max}+\lambda_T)/2$ and $e=(\lambda_{\max}-\lambda_T)/2$. We note that the ChFSI procedure for solving the eigenproblem in \cref{eqn:ghep} for the case of sparse matrices\footnote{For dense matrices the common practice is to employ Cholesky factorization of $\bB$ to convert the generalized eigenvalue problem to the standard eigenvalue problem.} is usually devised with the matrix $\bH=\bB^{-1}\bA$ that has the same eigenpairs as $\bA\bu_j=\lambda_j\bB\bu_j$.
This procedure for computing the smallest eigenpairs of \cref{eqn:ghep} up to a specified tolerance $\tau$ on the eigenproblem residual norm, is summarized in \cref{alg:ChFSI}. We note that in \cref{alg:ChFSI} we do not explicitly perform an orthogonalization step for the filtered subspace and instead opt to perform the Rayleigh-Ritz step for a non-orthogonal basis, which ensures that the resulting Ritz vectors are $\bB$-orthonormal. This choice was made for both ease of analysis and computational efficiency.

{\cblue The key computational kernel of ChFSI is the Chebyshev polynomial filter, which constructs the filtered subspace $\bY_p^{(i)}=C_p(\bH)\bX^{(i)}$ via the three-term recurrence
\begin{align}
  \bY_{k+1}^{(i)}=\frac{2\sigma_{k+1}}{e}\bH\bY_k^{(i)}-\frac{2\sigma_{k+1}c}{e}\bY_k^{(i)}-\sigma_k\sigma_{k+1}\bY_{k-1}^{(i)}\label{eqn:ChFSIrecurrence}
\end{align}
where $\bY_0^{(i)}=\bX^{(i)}$, $\bY_1^{(i)}=\frac{\sigma_1}{e}(\bH-c\bI)\bX^{(i)}$, and the scaling coefficients satisfy $\sigma_{k+1}=1/\left(\frac{2}{\sigma_1} - \sigma_k\right)$ with $\sigma_1=e/(a_{low}-c)$. Here $\bY_k^{(i)}=C_k(\bH)\bX^{(i)}$ for $k=0,1,\dots,p$. The filtering procedure is detailed in \cref{alg:chebFilt}. Following the filtering step, a Rayleigh--Ritz projection is performed. Denoting the conjugate transpose by~$^\dagger$, one solves the reduced $n\times n$ dense generalized eigenvalue problem 
\[
{\bY_p^{(i)}}^\dagger\bA\bY_p^{(i)}\bE^{(i+1)}={\bY_p^{(i)}}^\dagger\bB\bY_p^{(i)}\bE^{(i+1)}\bLam^{(i+1)},
\]
where $\bE^{(i+1)}$ is the eigenvector matrix and $\bLam^{(i+1)}=\mathrm{diag}(\epsilon_1^{(i+1)},\dots,\epsilon_n^{(i+1)})$ contains the Ritz values. The updated Ritz vectors are $\bX^{(i+1)}=\bY_p^{(i)}\bE^{(i+1)}$, which are $\bB$-orthonormal by construction. These two steps are repeated until the residual norm $r_j^{(i+1)}=\norm{\bA\bx_j^{(i+1)}-\epsilon_j^{(i+1)}\bB\bx_j^{(i+1)}}$ falls below a desired tolerance~$\tau$ for all $j=1,\dots,n$. The overall procedure is summarized in \cref{alg:ChFSI}.\cn}

\begin{algorithm}[!h]
\cblue
\caption{Subspace Iteration accelerated using Chebyshev polynomial of degree $p$ (ChFSI)}
\label{alg:ChFSI}
\begin{flushleft}
\textbf{INPUTS:} Hermitian matrix $\bA\in\mathbb{C}^{m\times m}$, SPD matrix $\bB\in\mathbb{C}^{m\times m}$, number of wanted eigenpairs $n$, Chebyshev polynomial degree $p$, spectral bounds $\lambda_{\min},\lambda_{\max},\lambda_T$, residual tolerance~$\tau$.\\
\textbf{OUTPUT:} Approximate eigenvectors $\bX^{(i+1)}\in\mathbb{C}^{m\times n}$ and eigenvalues $\bLam^{(i+1)}\in\mathbb{R}^{n\times n}$.
\end{flushleft}
\begin{algorithmic}
\STATE{\textbf{Step 0.} Choose an initial guess $\bX^{(0)} = \begin{bmatrix}
  \bx_1^{(0)} & \bx_2^{(0)} &\dots&\bx_n^{(0)}\end{bmatrix}\in\mathbb{C}^{m\times n}$.}
\WHILE{$\max_j\norm{\bA\bx^{(i+1)}_j-\epsilon_j^{(i+1)}\bB\bx^{(i+1)}_j} \geq \tau$}
\STATE{\textbf{Step 1.} \emph{Chebyshev filtered subspace construction}: Compute $\bY_p^{(i)}=C_p(\bH)\bX^{(i)}$ using the recurrence relation~\cref{eqn:ChFSIrecurrence} (see \cref{alg:chebFilt}).}
\STATE{\textbf{Step 2.} \emph{Rayleigh--Ritz projection}: Solve ${\bY_p^{(i)}}^\dagger\bA\bY_p^{(i)}\bE^{(i+1)}={\bY_p^{(i)}}^\dagger\bB\bY_p^{(i)}\bE^{(i+1)}\bLam^{(i+1)}$ for the Ritz values $\bLam^{(i+1)}$ and Ritz vectors $\bE^{(i+1)}$.}
\STATE{\textbf{Step 3.} \emph{Subspace rotation}: Set $\bX^{(i+1)}\gets\bY_p^{(i)}\bE^{(i+1)}$.}
\ENDWHILE
\end{algorithmic}
\cn\end{algorithm}
\begin{algorithm}[!htbp]
\cblue
\caption{Chebyshev filtering procedure for generalized Hermitian eigenvalue problems}
\label{alg:chebFilt}
\begin{flushleft}
    \textbf{INPUTS:} Chebyshev polynomial order $p$, estimates of the bounds of the eigenspectrum $\lambda_{max},\lambda_{min}$, estimate of the upper bound of the wanted spectrum $\lambda_{T}$ and the initial guess of eigenvectors $\bX^{(i)}$\\
    \textbf{OUTPUT:} The filtered subspace $\bY_p^{(i)}$\\
    \textbf{TEMPORARY VARIABLES:} $\bX$, $\bY$
\end{flushleft}

\begin{algorithmic}
\STATE{$e \gets \frac{\lambda_{max}-\lambda_{T}}{2}$; $c \gets \frac{\lambda_{max}+\lambda_{T}}{2}$; $\sigma \gets \frac{e}{\lambda_{min}-c}$; $\sigma_1 \gets \sigma$; $\gamma \gets \frac{2}{\sigma_1}$}
\STATE{$\bX \gets \bX^{(i)}$; $\bY \gets \frac{\sigma_1}{e}(\bH-c\bI)\bX^{(i)}$}
\FOR{$k\gets 2$ to $p$}
\STATE{$\sigma_2\gets \frac{1}{\gamma-\sigma}$}
\STATE{$\bX\gets\frac{2\sigma_2}{e}\bH\bY-\frac{2\sigma_2}{e}c\bY-\sigma\sigma_2\bX$}
\STATE{swap$\left(\bX,\bY\right)$; $\sigma = \sigma_2$}
\ENDFOR
\RETURN \upshape $\bY$
\end{algorithmic}
\cn\end{algorithm}

\subsubsection{Convergence analysis of ChFSI}
To ensure completeness, we now provide a mathematical justification for the convergence of ChFSI. {\cblue The convergence of subspace iteration accelerated by polynomial filters, including Chebyshev polynomials, has been studied in the context of general subspace iteration methods (see, e.g., \cite{watkinsMatrixEigenvalueProblem2007,saadNumericalMethodsLarge2011}). The analysis presented here follows the framework of Watkins~\cite{watkinsMatrixEigenvalueProblem2007}, adapted to explicitly track the maximum principal angle between the Chebyshev polynomial filtered subspace and the target eigenspace, and extended in \cref{sec:apprxSubspaceConstr} to account for approximation errors in the matrix--vector products---a setting not addressed in prior works.\cn} Note that $\norm{\cdot}$ implies the 2-norm for vectors and vector-induced matrix 2-norm for matrices (spectral norm) throughout this work. In order to analyze convergence, we first define the maximum principal angle between two subspaces~\cite{chenSpectralMethodsData2021,zhuAnglesSubspacesTheir2013,watkinsMatrixEigenvalueProblem2007}.
% \begin{definition}
%  Let $\mathcal{X}\subset\mathbb{C}^{m}$ and $\mathcal{Y}\subset\mathbb{C}^{m}$ be subspaces of dimensions $p$ and $q$ respectively. The maximum principal angle between the two subspaces denoted by $\angle(\mathcal{X},\mathcal{Y})$ is defined as 
%  \begin{align*}
%   \normalfont\tan\angle(\mathcal{X},\mathcal{Y})=\norm{\bX_\perp^\dagger\bY\left(\bX^\dagger\bY\right)^{+}}
%  \end{align*}
%  where $\normalfont\bX\in\mathbb{C}^{m\times p}$ is a matrix whose columns form an orthonormal basis for $\mathcal{X}$ and $\normalfont\bY \in \mathbb{C}^{m\times q}$ is a matrix whose columns form a basis, not necessarily orthogonal provided $q\leq p$, for $\mathcal{Y}$. $\normalfont\bX_\perp\in\mathbb{C}^{m\times (m-p)}$ is a matrix whose columns form an orthonormal basis spanning the orthogonal complements of $\mathcal{X}$. Here $\dagger$ denotes the transpose conjugate of a matrix and $+$ denotes the Moore-Penrose inverse of a matrix\label{def:angle}
% \end{definition}
\begin{definition}
\cblue Let $\mathcal{X},\mathcal{Y}\subset\mathbb{C}^{m}$ be subspaces of dimensions $p$ and $q$. 
Let $\normalfont \bX\in\mathbb{C}^{m\times p}$ have orthonormal columns spanning $\mathcal{X}$ and 
$\normalfont\bX_\perp\in\mathbb{C}^{m\times (m-p)}$ have orthonormal columns spanning $\normalfont\mathcal{X}^\perp$. 
Let $\normalfont\bY\in\mathbb{C}^{m\times q}$ have orthogonal columns spanning $\mathcal{Y}$. 
The maximum principal angle $\angle(\mathcal{X},\mathcal{Y})\in[0,\tfrac{\pi}{2}]$ is defined by
 \begin{align*}
\normalfont\tan\angle(\mathcal{X},\mathcal{Y})=\normalfont
\begin{cases}
\norm{\bX_\perp^{\dagger}\bY\left(\bX^{\dagger}\bY\right)^{+}}, 
& \text{if }\mathrm{rank}(\bX^{\dagger}\bY)=q, \\[6pt]
\frac{\pi}{2}, & \text{otherwise},
\end{cases}
 \end{align*}
where $\dagger$ denotes the conjugate transpose, $+$ the Moore--Penrose inverse.\cn\label{def:angle}
\end{definition}
For the sake of convenience in the subsequent convergence analysis, we define the Hermitian matrix $\hat{\bH}=\bB^{-\frac{1}{2}}\bA\bB^{-\frac{1}{2}}$ allowing us to rewrite the generalized eigenproblem $\bA\bU=\bB\bU\bLam$, as a Hermitian standard eigenvalue problem $\hat{\bH}\hat{\bU}=\hat{\bU}\bLam$, where $\hat{\bU}=\bB^{\frac{1}{2}}\bU$ denotes the $m\times m$ unitary matrix with columns as eigenvectors of $\hat{\bH}$ and $\bLam$ denotes the diagonal matrix comprising eigenvalues of $\hat{\bH}$. Consider the partitioning of $\hat{\bU}$ as $\begin{bmatrix}\hat{\bU}_1 & \hat{\bU}_2\end{bmatrix}$ where $\hat{\bU}_1$ is the $m\times n$ matrix whose columns are the eigenvectors corresponding to the lowest $n$ eigenvalues and $\hat{\bU}_2$ is the $m\times(m-n)$ matrix whose columns are the rest of the eigenvectors. We define the wanted eigenspace corresponding to the $n$ smallest eigenvectors of $\hat{\bH}$ as $\mathcal{S}=\mathcal{R}(\hat{\bU}_1)$ and the unwanted eigenspace as $\mathcal{S}_\perp=\mathcal{R}(\hat{\bU}_2)$. The trial subspace corresponding to $\hat{\bH}$ at the beginning of $i^{th}$ iteration is denoted as $\mathcal{S}^{(i)}=\mathcal{R}(\hat{\bX}^{(i)})$ where $\hat{\bX}^{(i)}=\bB^{\frac{1}{2}}\bX^{(i)}$. Further, the filtered subspace obtained at the end of the $i^{th}$ iteration is denoted as $\mathcal{S}^{(i+1)}=\mathcal{R}({\hat{\bY}_{p}}^{(i)})$ where $\hat{\bY}^{(i)}_p=\bB^{\frac{1}{2}}\bY^{(i)}_p$. We note that using the relation $\bB^{\frac{1}{2}}C_p(\bH)\bB^{-\frac{1}{2}}=C_p(\hat{\bH})$ and $\bY_p^{(i)}=C_p(\bH)\bX^{(i)}$ from \cref{alg:ChFSI}, we can conclude that $\hat{\bY}^{(i)}_p = C_p(\hat{\bH}) \hat{\bX}^{(i)}$, i.e., $\mathcal{S}^{(i+1)}=C_p(\hat{\bH})\mathcal{S}^{(i)}$. 
%\newpage
\begin{theorem}
 \label{thm:accurateConv}
 For an n-dimensional space $\mathcal{S}^{(i)}$ satisfying $\mathcal{S}^{(i)}\cap\mathcal{S}_\perp=\left\{0\right\}$ where $\mathcal{S}_\perp$ is the orthogonal complement of $\mathcal{S}$ and $\normalfont\mathcal{S}^{(i+1)}=C_p(\hat{\bH})\mathcal{S}^{(i)}$, we have the following inequality
 \begin{equation*}
 \normalfont\tan{\angle(\mathcal{S}^{(i+1)},\mathcal{S})}\leq\abs*{\frac{C_p(\lambda_{n+1})}{C_p(\lambda_n)}}\tan{\angle(\mathcal{S}^{(i)},\mathcal{S})}
 \end{equation*}
\end{theorem}
\begin{proof}
Note that $\hat{\bX}^{(i)}$ is an orthogonal matrix by virtue of implicit $\bB$ orthogonalization of $\bX^{(i)}$ in the Rayleigh-Ritz step. Define the $n\times n$ matrix $\hat{\bZ}_1^{(i)}=\hat{\bU}_1^\dagger\hat{\bX}^{(i)}$ and the $(m-n)\times n$ matrix $\hat{\bZ}_2^{(i)}=\hat{\bU}_2^\dagger\hat{\bX}^{(i)}$. Partitioning the matrix $\hat{\bU}^\dagger\hat{\bX}^{(i)}$, we have $\hat{\bU}^\dagger\hat{\bX}^{(i)}=\Bigl[\begin{smallmatrix}
  \hat{\bZ}_1^{(i)}\\\hat{\bZ}_2^{(i)}
 \end{smallmatrix}\Bigr]$. We note that the assumption $\mathcal{S}^{(i)}\cap\mathcal{S}_\perp=\left\{0\right\}$ ensures that $\hat{\bZ}_1^{(i)}$ is invertible. {\cblue This assumption simply requires that the current trial subspace $\mathcal{S}^{(i)}$ is not entirely contained in the unwanted eigenspace; in other words, the trial subspace must retain a nonzero component along every direction of the target eigenspace $\mathcal{S}$. For a randomly initialized subspace this holds with probability~one, and the Chebyshev filter strengthens the wanted component at each iteration, so the condition is maintained throughout the algorithm.\cn}
Using the maximum principal angle given by \cref{def:angle}, we have $\tan{\angle(\mathcal{S}^{(i)},\mathcal{S})}=\norm{\hat{\bZ}_2^{(i)}{{}\hat{\bZ}_1^{(i)}}^{+}}$. We can now write
\begin{align*}
 \mathcal{S}^{(i+1)}=\mathcal{R}(\hat{\bX}^{(i+1)})=\mathcal{R}(C_p(\hat{\bH})\hat{\bX}^{(i)}\bE^{(i+1)})
\end{align*}
We can now write $\hat{\bZ}_1^{(i+1)}=\hat{\bU}_1^\dagger\hat{\bX}^{(i+1)}$ and the $(m-n)\times n$ matrix $\hat{\bZ}_2^{(i+1)}=\hat{\bU}_2^\dagger\hat{\bX}^{(i+1)}$
\begin{align*}
 \tan{\angle(\mathcal{S}^{(i+1)},\mathcal{S})}\leq&\norm{C_p(\bLam_2)}\norm{\hat{\bZ}_2^{(i)}{{}\hat{\bZ}_1^{(i)}}^{+}}\norm{C_p(\bLam_1)^{-1}}=\frac{\max_{j\geq n+1}\abs{ C_p(\lambda_{j})}}{\min_{j\leq n}\abs{C_p(\lambda_j)}}\tan{\angle(\mathcal{S}^{(i)},\mathcal{S})}\\=&\abs*{\frac{C_p(\lambda_{n+1})}{C_p(\lambda_n)}}\tan{\angle(\mathcal{S}^{(i)},\mathcal{S})}
\end{align*}
\end{proof} 
\Cref{thm:accurateConv} demonstrates as the iterations progress, the Chebyshev filtered subspace approaches the wanted eigenspace.
\section{Approximate Chebyshev Filtered Subspace Construction}\label{sec:apprxSubspaceConstr}
We note that the computationally dominant step in Chebyshev filtered subspace construction is the evaluation of the matrix multi-vector product $\bH\bY_k^{(i)}$ in \cref{eqn:ChFSIrecurrence}. A straightforward way to accelerate this step is to use approximations in the computation of $\bH\bY_k^{(i)}$, allowing for improved efficiency. This should, in principle, allow for the use of various efficient approximate matrix multiplication techniques, including but not limited to mixed-precision arithmetic. However, we note that this requires an understanding of the convergence properties of ChFSI when such approximations are employed. We now adapt \cref{thm:accurateConv} for the case where approximations are used in computing $\bH\bY_k^{(i)}$.

\subsection{Convergence of ChFSI with inexact subspace construction}
We now define $\underline{\mathcal{S}}^{(i+1)}$ as the space spanned by the columns of $\hat{\underbar{\bY}}_p^{(i)}=\bB^{\frac{1}{2}}{\underbar {\bY}}_p^{(i)}$ with ${\underbar{\bY}}_p^{(i)}$ defined as $\underbar{\bY}_p^{(i)}=\apx{C_p(\bH)\bX^{(i)}}$ where underline here denotes that approximations are introduced during matrix multiplications in Step 2 of \cref{alg:ChFSI}. We note that the columns of $\hat{\bX}^{(i)}=\bB^{\frac{1}{2}}\bX^{(i)}$ form an orthonormal basis for $\mathcal{S}^{(i)}$. 

\begin{theorem}
\label{thm:apprxConv}
 For an n-dimensional space $\mathcal{S}^{(i)}$ satisfying $\mathcal{S}^{(i)}\cap\mathcal{S}_\perp=\left\{0\right\}$ and $\normalfont\underline{\mathcal{S}}^{(i+1)}=\mathcal{R}\left(\bB^{\frac{1}{2}}\apx{C_p(\bH)\bX^{(i)}}\right)$ we can write
 \begin{align}
   \tan{\angle(\underline{\mathcal{S}}^{(i+1)},\mathcal{S})}&\leq\left(\frac{\abs{C_p(\lambda_{n+1})}+\norm{\hat{\bDelta}_p^{(i)}}\csc{\angle(\mathcal{S}^{(i)},\mathcal{S})}}{\abs{C_p(\lambda_n)}-\norm{\hat{\bDelta}_p^{(i)}}\sec{\angle(\mathcal{S}^{(i)},\mathcal{S})}}\right)\tan{\angle(\mathcal{S}^{(i)},\mathcal{S})}\label{eqn:convIneqApp}
 \end{align}
 where $\normalfont\hat{\bDelta}_p^{(i)}=\bB^{\frac{1}{2}}\bDelta_p=\bB^{\frac{1}{2}}(\apx{C_p(\bH)\bX^{(i)}}-C_p(\bH)\bX^{(i)})$ 
\end{theorem}
\begin{proof}
Note that $\hat{\bX}^{(i)}$ is an orthogonal matrix by virtue of implicit $\bB$ orthogonalization of $\bX^{(i)}$ in the Rayleigh-Ritz step. Define the $n\times n$ matrix $\hat{\bZ}_1^{(i)}=\hat{\bU}_1^\dagger\hat{\bX}^{(i)}$ and the $(m-n)\times n$ matrix $\hat{\bZ}_2^{(i)}=\hat{\bU}_2^\dagger\hat{\bX}^{(i)}$. Partitioning the matrix $\hat{\bU}^\dagger\hat{\bX}^{(i)}$, we have $\hat{\bU}^\dagger\hat{\bX}^{(i)}=\Bigl[\begin{smallmatrix}
  \hat{\bZ}_1^{(i)}\\\hat{\bZ}_2^{(i)}
 \end{smallmatrix}\Bigr]$. We note that the assumption $\mathcal{S}^{(i)}\cap\mathcal{S}_\perp=\left\{0\right\}$ ensures that $\hat{\bZ}_1^{(i)}$ is invertible.
Using the maximum principal angle given by \cref{def:angle}, we have $\tan{\angle(\mathcal{S}^{(i)},\mathcal{S})}=\norm{\hat{\bZ}_2^{(i)}{{}\hat{\bZ}_1^{(i)}}^{+}}$. In the eigenvector coordinate system, the approximate filtered subspace can be written as
 \begin{align*}
  \hat{\bU}^\dagger\underline{\mathcal{S}}^{(i+1)}&=\mathcal{R}\left(\hat{\bU}^\dagger C_p(\hat{\bH})\hat{\bX}^{(i)}+\hat{\bU}^\dagger\hat{\bDelta}_p^{(i)}\right)=\mathcal{R}\left(\begin{bmatrix}
   C_p(\bLam_1)\hat{\bZ}_1^{(i)}+\hat{\bU}_1^\dagger\hat{\bDelta}_p^{(i)}\\C_p(\bLam_2)\hat{\bZ}_2^{(i)}+\hat{\bU}_2^\dagger\hat{\bDelta}_p^{(i)}
  \end{bmatrix}\right)
 \end{align*}
Using \cref{def:angle}, we can write
 \begin{align}
\tan{\angle(\underline{\mathcal{S}}^{(i+1)},\mathcal{S})}&=\norm*{\left(C_p(\bLam_2)\hat{\bZ}_2^{(i)}+\hat{\bU}_2^\dagger\hat{\bDelta}_p^{(i)}\right)\left(C_p(\bLam_1)\hat{\bZ}_1^{(i)}+\hat{\bU}_1^\dagger\hat{\bDelta}_p^{(i)}\right)^{+}} \nonumber\\
  &\leq\frac{\norm{C_p(\bLam_2)\hat{\bZ}_2^{(i)}{{}\hat{\bZ}_1^{(i)}}^{+}C_p(\bLam_1)^{-1}}+\norm{\hat{\bDelta}_p^{(i)}{{}\hat{\bZ}_1^{(i)}}^{+}C_p(\bLam_1)^{-1}}}{1-\norm{\hat{\bDelta}_p^{(i)}{{}\hat{\bZ}_1^{(i)}}^{+}C_p(\bLam_1)^{-1}}} \label{eqn:tanangleapprox}
 \end{align}
 Note that in the last step, we have assumed that $\norm{\hat{\bDelta}_p^{(i)}{{}\hat{\bZ}_1^{(i)}}^{+}C_p(\bLam_1)^{-1}}<1$, and using the fact that $\norm{{{}\hat{\bZ}_1^{(i)}}^{+}}=\sec{\angle(\mathcal{S}^{(i)},\mathcal{S})}$, a sufficient condition for this to be true is $\abs{C_p(\lambda_n)}\cos{\angle(\mathcal{S}^{(i)},\mathcal{S})}>\norm{\hat{\bDelta}_p^{(i)}}$. Upon further simplification using sub-multiplicative property of matrix spectral norms and triangle inequality, the inequality in \cref{eqn:tanangleapprox} can be written as 
 
 \begin{align*}
  \tan{\angle(\underline{\mathcal{S}}^{(i+1)},\mathcal{S})}&\leq\frac{\abs{C_p(\lambda_{n+1})}\sin{\angle(\mathcal{S}^{(i)},\mathcal{S})}+\norm{\hat{\bDelta}_p^{(i)}}}{\abs{C_p(\lambda_n)}\cos{\angle(\mathcal{S}^{(i)},\mathcal{S})}-\norm{\hat{\bDelta}_p^{(i)}}}\\
  &=\left(\frac{\abs{C_p(\lambda_{n+1})}+\norm{\hat{\bDelta}_p^{(i)}}\csc{\angle(\mathcal{S}^{(i)},\mathcal{S})}}{\abs{C_p(\lambda_n)}-\norm{\hat{\bDelta}_p^{(i)}}\sec{\angle(\mathcal{S}^{(i)},\mathcal{S})}}\right)\tan{\angle(\mathcal{S}^{(i)},\mathcal{S})}
 \end{align*}
which proves the desired inequality in the theorem.
\end{proof}
Now for convergence we demand that $\angle(\underline{\mathcal{S}}^{(i+1)},\mathcal{S})<\angle(\mathcal{S}^{(i)},\mathcal{S})$ and consequently we require
\begin{equation}
 \abs{C_p(\lambda_n)}-\abs{C_p(\lambda_{n+1})}\!>\!\norm{\hat{\bDelta}_p^{(i)}}\left(\sec{\angle(\mathcal{S}^{(i)},\mathcal{S})}+\csc{\angle(\mathcal{S}^{(i)},\mathcal{S})}\right) \; \forall i=0,\dots,\infty\label{eqn:convIneqAppReq}
\end{equation}
We note that if $\norm{\hat{\bDelta}_p^{(i)}}$ remains nearly constant with iteration $i$, the right-hand side of the inequality in \cref{eqn:convIneqAppReq} keeps increasing as we approach the exact eigenspace and beyond a certain angle $\angle(\mathcal{S}^{(i)},\mathcal{S})$ this inequality gets violated. Consequently, the angle stops decreasing, and we cannot approach the exact eigenspace beyond that point. To this end, for robust convergence, we require that $\norm{\hat{\bDelta}_p^{(i)}}$ also decreases as we approach the exact eigenspace, and we demonstrate that our proposed residual-based reformulation of Chebyshev filtered subspace iteration procedure (R-ChFSI) algorithm described subsequently accomplishes this. We further note that $\left(\sec{\angle(\mathcal{S}^{(i)},\mathcal{S})}+\csc{\angle(\mathcal{S}^{(i)},\mathcal{S})}\right)\geq 2\sqrt{2}$ and consequently for convergence with approximations, we obtain the following necessary condition on $\norm{\hat{\bDelta}_p^{(i)}}$ from \cref{eqn:convIneqAppReq}, \begin{align}
 \norm{\hat{\bDelta}_p^{(i)}}<\frac{\abs{C_p(\lambda_n)}-\abs{C_p(\lambda_{n+1})}}{2\sqrt{2}}.
\end{align}
\subsection{Error in subspace construction}
In this section we consider the specific case of utilizing a lower precision to compute the matrix product $\bH\bY^{(i)}_k$ in conjunction with approximating $\bB^{-1}$ with $\bD^{-1}$ in the Chebyshev recurrence relation defined by \cref{eqn:ChFSIrecurrence}. We first evaluate the upper bound on $\bDelta_p^{(i)}$ if one naively replaces the matrix product $\bH\bY^{(i)}_k$ in \cref{eqn:ChFSIrecurrence} with the approximate matrix product denoted as ${\bD^{-1}\otimes\bA\otimes\bY^{(i)}_k}$, where $\otimes$ represents the product evaluated with lower precision arithmetic, and argue that this method will fail to converge to the same residual tolerance that can be achieved with full precision matrix products and no approximation in computing the action of $\bB^{-1}$. We then propose a residual-based reformulation of the recurrence relation and argue that the proposed reformulation allows for convergence to similar residual tolerances that can be achieved with full precision matrix products and without approximations.

\subsubsection{Traditional Chebyshev filtering method}
Employing low-precision matrix-products and approximate inverse, the recurrence relation in \cref{eqn:ChFSIrecurrence} for $\underline{\bY}^{(i)}_k$, where the underline denotes that the matrix was constructed using approximations and $k=2,\dots,p$, can be written as
\begin{align}
 \underline{\bY}_{k+1}^{(i)}=a_k{\bD^{-1}\otimes\bA\otimes\underline{\bY}_k^{(i)}}+b_k\underline{\bY}_k^{(i)}+c_k\underline{\bY}_{k-1}^{(i)}\label{eqn:ChFSIrecurrenceSEPMP}
\end{align}
where for convenience of notation we have defined
\begin{align*}
  a_k&=\frac{2\sigma_{k+1}}{e}&b_k&=-\frac{2\sigma_{k+1}c}{e}&c_k&=-\sigma_k\sigma_{k+1}
\end{align*}
and with the initial conditions $\underline{\bY}_0^{(i)}=\bX^{(i)}$ and $\underline{\bY}_1^{(i)}=\frac{\sigma_1}{e}(\bD^{-1}\bA-c\bI)\bX^{(i)}$. 

\begin{theorem}
 \label{thm:deltaNChFSI}
 The spectral norm of the error $\normalfont\hat{\bDelta}_k^{(i)}=\bB^{\frac{1}{2}}(\underline{\bY}_k^{(i)}-\bY_k^{(i)})$ in the subspace construction using inexact matrix-products of the recurrence relation \cref{eqn:ChFSIrecurrenceSEPMP} satisfies $\normalfont\norm{\hat{\bDelta}_k^{(i)}}\leq\gamma_m\eta_k+\zeta\tilde{\eta}_k$ for $k=0,1,\dots,p$ where $\eta_k$ and $\tilde{\eta}_k$ are some finite constants that depend on $k$ and $\normalfont\zeta=\norm{\bD^{-1}-\bB^{-1}}$
\end{theorem}
\begin{proof}
This is proved in \cref{prf:deltaNChFSI} of the Appendix.
\end{proof}

We now discuss the implications of the \cref{thm:deltaNChFSI} within the context of \cref{eqn:convIneqAppReq}. Without loss of generality, we assume that our initial guess of trial subspace $\mathcal{S}^{(i)}$ at $i=0$ satisfies the inequality below as stated in \cref{eqn:convIneqAppReq}.
\begin{align}
 \abs{C_p(\lambda_n)}-\abs{C_p(\lambda_{n+1})}&>\norm{\hat{\bDelta}_p^{(i)}}\left(\sec{\angle(\mathcal{S}^{(i)},\mathcal{S})}+\csc{\angle(\mathcal{S}^{(i)},\mathcal{S})}\right) \label{eqn:convIneqAppReq1}
\end{align}

While the above inequality is satisfied, $\angle(\mathcal{S}^{(i)},\mathcal{S})$ reduces as $i$ increases (from Theorem 3.1). However, we note that the RHS of this equation does not monotonically decrease with a decrease in $\angle(\mathcal{S}^{(i)},\mathcal{S})$ and increases as $\angle(\mathcal{S}^{(i)},\mathcal{S})$ approaches 0. Thus, it stands to reason that beyond a certain value of $\angle(\mathcal{S}^{(i)},\mathcal{S})$, the inequality in \cref{eqn:convIneqAppReq1} no longer holds, and we cannot approach the wanted eigenspace beyond this point. Under the condition that $\gamma_m\eta_{p}+\zeta\tilde{\eta}_p<<\abs{C_p(\lambda_n)}-\abs{C_p(\lambda_{n+1})}$, we can estimate the closest angle that can be achieved and is given by the following expression
\begin{align}
\angle(\mathcal{S}^{(i)},\mathcal{S})\approx\frac{\gamma_m\eta_{p}+\zeta\tilde{\eta}_p}{\abs{C_p(\lambda_n)}-\abs{C_p(\lambda_{n+1})}}
\end{align}
We now propose a residual-based reformulation of ChFSI and argue that the proposed method does not suffer from this stagnation behavior even when employing lower precision arithmetic in matrix products.
 % What will be delta in naive way and why it won't converge
\subsubsection{Proposed residual-based Chebyshev filtering approach}
To this end we define the weighted residual $\bZ^{(i)}_{k}$ in a given iteration $i$ for $k=0,\dots,p$ in the following way:
\begin{equation}
  \bZ_k^{(i)}=\bD\bR_k^{(i)}=\bD(C_k(\bH)\bX^{(i)}-\bX^{(i)} C_k(\bLam^{(i)}))=\bD(\bY_k^{(i)}-\bX^{(i)}\bLam_k^{(i)})\;\;\; \text{for} \;\;\; k=0,\dots,p \label{eqn:residual}
\end{equation}
where we have defined $\bLam_k^{(i)}=C_k(\bLam^{(i)})$ and recall $p$ is the maximum Chebyshev polynomial degree used in the subspace construction step. The motivation behind this definition is the fact that both $C_k(\bH)$ and $\bH$ share the same eigenvectors with the eigenvalues being $C_k(\bLam)$ and $\bLam$ respectively, this ensures that as $\bX^{(i)}$ and $\bLam^{(i)}$ approach the exact eigenvectors and eigenvalues the weighted residual $\bZ_k^{(i)}$ approaches zero. Now, if we rewrite the recurrence relation defined in \cref{eqn:ChFSIrecurrence} such that the matrix products are in terms of $\bZ_k^{(i)}$ then the relative error in the matrix products is then proportional to $\norm{\bZ_k^{(i)}}$ thus ensuring a lower absolute error in the filtered subspace. We now propose the following recurrence relation that accomplishes our objective
\begin{proposition}
 The recurrence relation given by \cref{eqn:ChFSIrecurrence} can be reformulated in terms of the weighted residuals defined by $\normalfont\bZ_k^{(i)}=\bD\bR_k^{(i)}=\bD(\bY_k^{(i)}-\bX^{(i)}\bLam_k^{(i)})$ as
 \begin{align}
  \normalfont\bZ_{k+1}^{(i)}&\normalfont=a_k\bD\bH\bD^{-1}\bZ_k^{(i)}+b_k\bZ_k^{(i)}+c_k\bZ_{k-1}^{(i)}+a_k\bD\bR^{(i)}\bLam_k^{(i)}\label{eqn:rChFSIrecurrence}\\
  \bLam_{k+1}^{(i)}&=a_k\bLam_k^{(i)}\bLam^{(i)}+b_k\bLam_k^{(i)}+c_k\bLam_{k-1}^{(i)}\label{eqn:ChFSIrecurrenceEigval}
 \end{align} 
where $\normalfont\bZ_0^{(i)}=0$ and $\normalfont\bZ_1^{(i)}=\frac{\sigma_{1}}{e}\bD\bR^{(i)}$ with $\normalfont\bR^{(i)}=\bH\bX^{(i)}-\bX^{(i)}\bLam^{(i)}$ and further we have $a_k=2\sigma_{k+1}/e$, $b_k=-2\sigma_{k+1}c/e$, and $c_k=-\sigma_k\sigma_{k+1}$.
\end{proposition}
\begin{proof}
 The recurrence relation for $\bLam_k^{(i)}=C_k(\bLam^{(i)})$ can be written as 
 \begin{align}
  \bLam_{k+1}^{(i)}=a_k\bLam_k^{(i)}\bLam^{(i)}+b_k\bLam_k^{(i)}+c_k\bLam_{k-1}^{(i)}\nonumber
 \end{align}
 Multiplying with $\bX^{(i)}$, subtracting from \cref{eqn:ChFSIrecurrence}, and denoting $\bZ_k^{(i)}=\bD\bR_k^{(i)}$ we have
 \begin{align*}
  \bZ_{k+1}^{(i)}&=a_k\bD\bH\bD^{-1}\bZ_k^{(i)}+b_k\bZ_k^{(i)}+c_k\bZ_{k-1}^{(i)}+a_k\bD\bR^{(i)}\bLam_k^{(i)}
 \end{align*}
\end{proof}
After computing $\bZ_p^{(i)}$ using this recurrence relation we can now evaluate $\bY_p^{(i)}$ using the relation $\bY_p^{(i)}=\bD^{-1}\bZ_p^{(i)}+\bX^{(i)}\bLam_p^{(i)}$. We note that using the approximation $\bD\bB^{-1}\approx\bI$ and employing lower precision matrix products in \cref{eqn:rChFSIrecurrence}, we can write the following recurrence relation
\begin{align}
 \underline{\bZ}_{k+1}^{(i)}=a_k{\bA\otimes\bD^{-1}\otimes\underline{\bZ}_k^{(i)}}+b_k\underline{\bZ}_k^{(i)}+c_k\underline{\bZ}_{k-1}^{(i)}+a_k\bB\bR^{(i)}\bLam_k^{(i)}\label{eqn:rChFSIrecurrenceMP}
\end{align}
with $\bZ_0^{(i)}=0$ and $\bZ_1^{(i)}=\frac{\sigma_{1}}{e}\bB\bR^{(i)}$ and consequently, we have $\underline{\bY}_k^{(i)}=\bD^{-1}\underline{\bZ}_k^{(i)}+\bX^{(i)}\bLam_k^{(i)}$. Note that this recurrence relation does not require the evaluation of $\bB^{-1}$.

\begin{theorem}
The spectral norm of the error $\normalfont\bDelta_k^{(i)}=\underline{\bY}_k^{(i)}-\bY_k^{(i)}$ in the subspace construction using recurrence relation given by \cref{eqn:rChFSIrecurrenceMP} satisfies $\normalfont\norm{\bDelta_k^{(i)}}\leq(\gamma_m\eta_k+\zeta\tilde{\eta}_k)\norm{\bR^{(i)}}$ for $k=0,1,\dots,p$ where $\eta_k$ and $\tilde{\eta}_k$ are some finite constants that depends on $k$ and $\normalfont\zeta=\norm{\bD^{-1}-\bB^{-1}}$.\label{thm:dkboundGen}
\end{theorem}
\begin{proof}
This is proved in \cref{prf:dkboundGen} of the Appendix.
\end{proof}

\begin{theorem}
The necessary condition for $\normalfont\angle(\underline{\mathcal{S}}^{(i+1)},\mathcal{S})<\angle(\mathcal{S}^{(i)},\mathcal{S})$ for the case of residual-based Chebyshev filtering approach can be written as
\begin{align}
 \normalfont\abs{C_p(\lambda_n)}-\abs{C_p(\lambda_{n+1})}&\normalfont >2\norm{\bH}(\gamma_m\eta_p+\zeta\tilde{\eta}_p)\left(1+\tan{\angle(\mathcal{S}^{(i)},\mathcal{S})}\right)
\end{align}
and if this inequality is satisfied for $i=i_0$ then it holds for all $i>i_0$
\end{theorem}
\begin{proof}
From \cref{lem:rkbound,thm:dkboundGen,lem:TwoNormBound} we have $\norm{\hat{\bDelta}_p^{(i)}}\leq2\norm{\bH}(\gamma_m\eta_p+\zeta\tilde{\eta}_p)\sin{\angle(\mathcal{S}^{(i)},\mathcal{S})}$ and in order to have $\angle(\underline{\mathcal{S}}^{(i+1)},\mathcal{S})<\angle(\mathcal{S}^{(i)},\mathcal{S})$ from \cref{eqn:convIneqApp} we require
\begin{align}
 \abs{C_p(\lambda_n)}-\abs{C_p(\lambda_{n+1})}&>2\norm{\bH}(\gamma_m\eta_p+\zeta\tilde{\eta}_p)\left(1+\tan{\angle(\mathcal{S}^{(i)},\mathcal{S})}\right)
\end{align}
We note that if this inequality is satisfied for some $i=i_0$, then it holds for all $i>i_0$ as the RHS decreases with decreasing $\angle(\mathcal{S}^{(i)},\mathcal{S})$ and hence we conclude that the R-ChFSI method can converge under approximations where the ChFSI method fails.
\end{proof}

Finally, the Chebyshev filtering step in R-ChFSI can be summarized in the following algorithm. 
\begin{algorithm}[!htbp]
\cblue
\caption{Residual based Chebyshev filtering procedure for generalized Hermitian eigenvalue problems}
\label{alg:RChFSiGhep}
\begin{flushleft}
    \textbf{INPUTS:} Chebyshev polynomial order $p$, estimates of the bounds of the eigenspectrum $\lambda_{max},\lambda_{min}$, estimate of the upper bound of the wanted spectrum $\lambda_{T}$, the initial guess of eigenvectors $\bX^{(i)}\in\mathbb{C}^{m\times n}$, current eigenvalue estimates $\bLam^{(i)}=\mathrm{diag}(\epsilon_1^{(i)},\dots,\epsilon_n^{(i)})$, approximate inverse of $\bB$ denoted $\bD^{-1}$, and the residual $\bY=\bA\bX^{(i)}-\bB\bX^{(i)}\bLam^{(i)}$.\\
    \textbf{OUTPUT:} The filtered subspace $\bY_p^{(i)}\in\mathbb{C}^{m\times n}$\\
    \textbf{TEMPORARY VARIABLES:} $\bR_{\bX}\in\mathbb{C}^{m\times n}$, $\bR_{\bY}\in\mathbb{C}^{m\times n}$ (filtered residual vectors), $\bLam_{\bX}\in\mathbb{R}^{n\times n}$, $\bLam_{\bY}\in\mathbb{R}^{n\times n}$ (Chebyshev-filtered eigenvalue matrices)
\end{flushleft}

\begin{algorithmic}
\STATE{$e \gets \frac{\lambda_{max}-\lambda_{T}}{2}$; $c \gets \frac{\lambda_{max}+\lambda_{T}}{2}$; $\sigma \gets \frac{e}{\lambda_{min}-c}$; $\sigma_1 \gets \sigma$; $\gamma \gets \frac{2}{\sigma_1}$}
\STATE{$\bX \gets \bX^{(i)}$; $\bY \gets \bA\bX^{(i)}-\bB\bX^{(i)}\bLam^{(i)}$}
\STATE{$\bR_{\bX}\gets 0$; $\bR_{\bY}\gets \frac{\sigma_{1}}{e}\bY$}
\STATE{$\bLam_{\bX} \gets \bI$; $\bLam_{\bY} \gets \frac{\sigma_{1}}{e}\left(\bLam^{(i)}-c\bI\right)$}
\FOR{$k\gets 2$ to $p$}
\STATE{$\sigma_2\gets \frac{1}{\gamma-\sigma}$}
\STATE{$\bR_{\bX}\gets\frac{2\sigma_2}{e}\bA\bD^{-1}\bR_{\bY}-\frac{2\sigma_2}{e}c\bR_{\bY}-\sigma\sigma_2\bR_{\bX}+\frac{2\sigma_2}{e}\bY\bLam_{\bY}$}
\STATE{$\bLam_{\bX}\gets \frac{2\sigma_2}{e}\bLam_{\bY}\bLam^{(i)}-\frac{2\sigma_2}{e}c\bLam_{\bY}-\sigma\sigma_2\bLam_{\bX}$}
\STATE{swap$\left(\bR_{\bX},\bR_{\bY}\right)$; swap$\left(\bLam_{\bX},\bLam_{\bY}\right)$; $\sigma = \sigma_2$}
\ENDFOR
\STATE{$\bX \gets\bD^{-1}\bR_{\bY}+\bX\bLam_{\bY}$}
\RETURN \upshape $\bX$
\end{algorithmic}
\cn\end{algorithm}

We would also like to comment that while our above mathematical analysis is focused on generalized eigenvalue problems, it is trivial to extend R-ChFSI to standard eigenvalue problems by setting $\bD=\bB=\bI$. When $\bD^{-1}=\bB^{-1}$ (i.e.\ the approximate inverse is exact) \emph{and} the same matrix is used for both the Chebyshev filter and the Rayleigh--Ritz projection, ChFSI and R-ChFSI are algebraically equivalent. However, when the matrix--vector product in the filter is inexact---e.g.\ the filter applies an approximate $\tilde{\bA}$ while the Rayleigh--Ritz step uses the exact~$\bA$---R-ChFSI retains its robustness advantage even in the standard eigenproblem. The same advantage arises from a mismatch $\bD^{-1}\neq\bB^{-1}$ in the generalized eigenproblem, which is precisely the scenario encountered in practice where forming an exact $\bB^{-1}$ is prohibitively expensive. In the interest of brevity, the subsequent section of this manuscript will focus on controlled experiments validating these claims on dense random matrices of moderate size, followed by solutions for generalized eigenproblems from DFT using both low-precision and approximate inverses in computing matrix products for constructing the subspace rich in the desired subspace using Chebyshev filtering.
\section{Results and Discussion}
We present a detailed evaluation of the accuracy and efficiency achieved by the proposed residual-based Chebyshev filtered subspace iteration (R-ChFSI) method for solving real symmetric and complex Hermitian generalized eigenvalue problems. We compare the performance of the R-ChFSI method (Algorithm 3) with that of the traditional ChFSI approach (Algorithm 2), demonstrating its improved accuracy even when employing inexact matrix-vector products. Additionally, we investigate the computational performance of our proposed method on the Intel Data Center GPU Max Series, utilizing the Aurora supercomputing system, demonstrating its ability to maintain high accuracy while achieving greater computational efficiency. {\cblue We first validate the theoretical claims of \cref{sec:apprxSubspaceConstr} on small dense eigenvalue problems with precisely controlled approximation errors, and then demonstrate the practical effectiveness of R-ChFSI on large-scale sparse eigenproblems arising in Kohn-Sham DFT.

\subsection{Controlled numerical experiments on dense random matrices}\label{sec:toyExperiments}
Before evaluating R-ChFSI on application-scale eigenproblems, we first isolate the effect of inexact matrix--vector products on the convergence of ChFSI and R-ChFSI in a fully controlled, reproducible setting. To this end we construct a family of dense eigenvalue problems of moderate size ($m=1000$) whose eigenvalue spectra, condition numbers, and approximation errors are all prescribed exactly. Since all parameters are explicitly controlled, the experiments allow a direct, quantitative verification of the convergence bounds derived in \cref{sec:apprxSubspaceConstr}, independent of any application-specific details.

\subsubsection{Problem construction and experimental setup}\label{sec:toySetup}

\paragraph{Spectrum of~$\bA$}
We generate an $m\times m$ ($m=1000$) real symmetric matrix~$\bA$ with a prescribed eigenvalue spectrum. Towards this, a random orthogonal matrix $\bQ$ is first obtained by computing the QR factorization (via Householder reflections, LAPACK routines \texttt{dgeqrf}/\texttt{dorgqr}) of an $m\times m$ matrix whose entries are drawn i.i.d.\ from $\mathcal{N}(0,1)$ with a fixed random seed for reproducibility. We then set $\bA = \bQ\,\mathrm{diag}(\lambda_1,\dots,\lambda_m)\,\bQ^\top$, where the $n=10$ \emph{wanted} eigenvalues are uniformly spaced in the interval $[1,4]$, i.e.\ $\lambda_j = 1 + 3(j-1)/(n-1)$ for $j=1,\dots,n$, and the remaining $m-n=990$ \emph{unwanted} eigenvalues are $\lambda_j = 5 + 0.2\,(j-n-1)$ for $j=n+1,\dots,m$. The maximal eigenvalue is therefore $\lambda_m \approx 202.8$. This construction yields a spectral gap of $\delta = \lambda_{n+1} - \lambda_n = 5.0 - 4.0 = 1.0$ between the wanted and unwanted spectra. The 2-norm condition number of~$\bA$ is $\kappa_2(\bA)=\lambda_m/\lambda_1\approx 203$, which is moderate and ensures that any convergence difficulties observed are attributable to the injected noise rather than ill-conditioning of the matrix itself. The Chebyshev polynomial filter of degree $p=8$ is applied with bounds $\lambda_{\min}=\lambda_1 - 0.05=0.95$, $\lambda_{\mathrm{T}} = (\lambda_n + \lambda_{n+1})/2 = 4.5$, and $\lambda_{\max}=\lambda_m + 0.1$. In all experiments we iterate for a maximum of $100$ outer iterations without early stopping, so that the full convergence history is recorded.

\paragraph{Standard eigenproblem: noise model.}
{\cblue For the standard eigenproblem ($\bB=\bI$), the matrix is perturbed directly: we form $\tilde{\bA}=\bA+\varepsilon\,\bE$, where $\bE$ is a random symmetric matrix normalized so that $\norm{\bE}_2=1$ and $\varepsilon\geq 0$ controls the perturbation magnitude. The approximate inverse is set to $\bD^{-1}=\bI$ (exact). During the Chebyshev filtering step the matrix--vector products use the perturbed~$\tilde{\bA}$, while the Rayleigh--Ritz projection and residual norms are evaluated with the exact~$\bA$. This models scenarios in which the matrix--vector product in the polynomial filter is executed in lower precision---for example, using mixed-precision arithmetic---while the Rayleigh--Ritz step retains full-precision arithmetic. We consider four perturbation levels, $\varepsilon\in\{0,\,10^{-4},\,10^{-3},\,10^{-2}\}$.\cn}

\paragraph{Generalized eigenproblem: construction of~$\bB$ and approximate inverse.}
For the generalized eigenproblem $\bA\bx=\lambda\bB\bx$, we additionally construct a symmetric positive-definite matrix~$\bB$ using a \emph{second} random orthogonal basis $\bQ_B$ (same seed as $\bQ$, yielding a shared eigenbasis so that the generalized eigenvalues are exactly $\lambda_j^{(\mathrm{gen})}=a_j/b_j$). The eigenvalues of~$\bB$ are linearly spaced in the interval $[1,5]$, i.e.\ $b_j = 1 + 4(j-1)/(m-1)$ for $j=1,\dots,m$, giving $\kappa_2(\bB)=b_m/b_1=5$. The exact inverse is then $\bB^{-1}=\bQ_B\,\mathrm{diag}(1/b_1,\dots,1/b_m)\,\bQ_B^\top$, and the approximate inverse is formed as
\[
\bD^{-1} = \bB^{-1} + \zeta\,\bE, \qquad \bE = \tfrac{1}{2}(\bF+\bF^\top),\quad F_{ij}\sim\mathcal{N}(0,1),
\]
where $\bE$ is a random symmetric perturbation normalized so that $\norm{\bE}_2 = 1$, and $\zeta\geq 0$ controls the perturbation magnitude. By construction, $\norm{\bD^{-1}-\bB^{-1}}_2 = \zeta$ exactly. This construction directly models the scenario encountered in practice when $\bB^{-1}$ is replaced by a lumped-mass or incomplete-factorization approximations. We test four levels, $\zeta\in\{0,\,10^{-4},\,10^{-3},\,10^{-2}\}$.

\subsubsection{Standard eigenproblem with perturbed matrix (\texorpdfstring{$\bB=\bI$}{B=I})}\label{sec:toyStd}
{\cblue \Cref{fig:ToyStandard} shows the convergence of ChFSI and R-ChFSI for the standard eigenproblem with matrix perturbation levels $\varepsilon\in\{0,\,10^{-4},\,10^{-3},\,10^{-2}\}$. Without any perturbation ($\varepsilon=0$) both methods converge identically, as expected from their algebraic equivalence when the same matrix is used everywhere.  For every nonzero perturbation level, however, the two methods diverge sharply: ChFSI stagnates at a subspace angle and residual of $\mathcal{O}(\varepsilon)$ once the Chebyshev filter can no longer distinguish the wanted subspace from the noise introduced by~$\tilde{\bA}$, while R-ChFSI converges to the exact eigenvectors of the unperturbed~$\bA$ down to machine precision.  The explanation mirrors the generalized-problem analysis of \cref{sec:apprxSubspaceConstr}: in the residual-based recurrence, the noise from~$\tilde{\bA}$ enters proportional to the current residual $\norm{\bR^{(i)}}$, which vanishes during convergence, so the filtering error decays in tandem with the residual.  In the traditional recurrence, the same noise is applied to the eigenvectors themselves, whose norm remains $\mathcal{O}(1)$, producing a persistent $\mathcal{O}(\varepsilon)$ error.\cn}

\begin{figure}[!ht]
 \includegraphics[width=\textwidth]{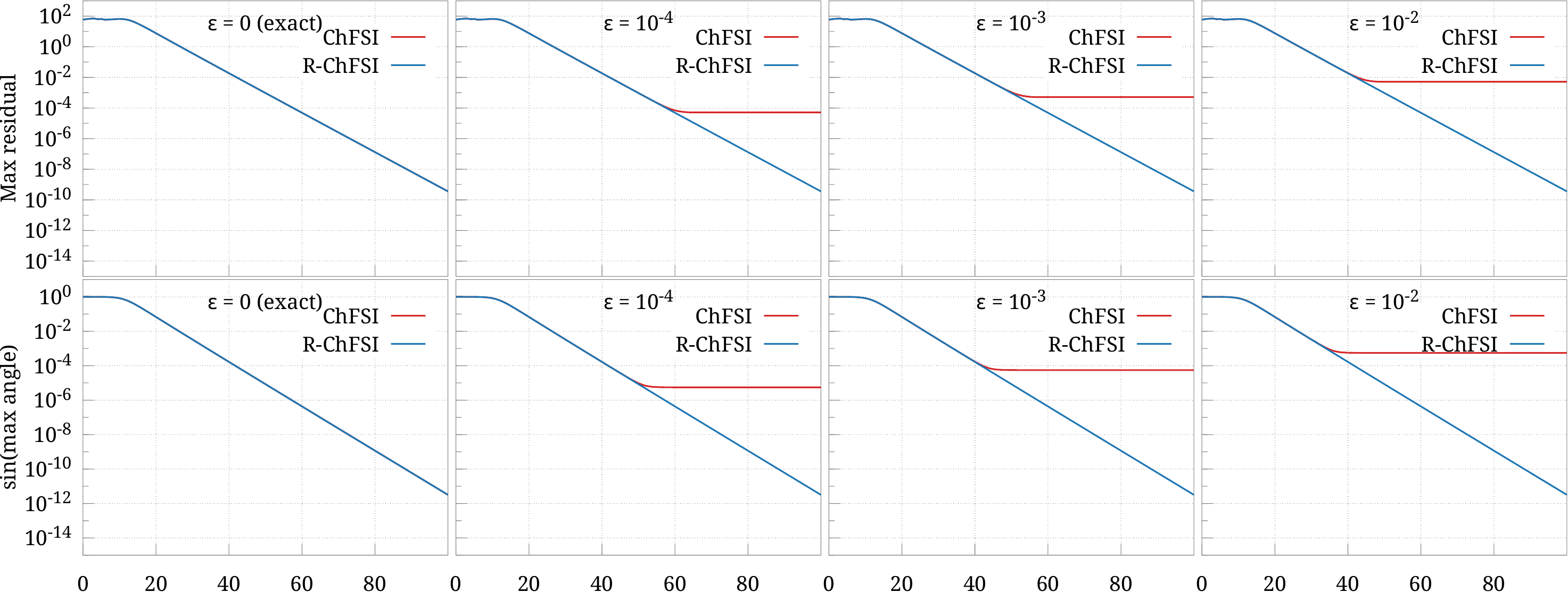}
 \caption{\cblue Convergence of ChFSI and R-ChFSI for the standard eigenproblem ($\bB=\bI$, $m=1000$, $n=10$, $p=8$, $\kappa_2(\bA)\approx 203$, spectral gap $\delta=1$) with matrix perturbation levels $\varepsilon=\norm{\tilde{\bA}-\bA}_2\in\{0,10^{-4},10^{-3},10^{-2}\}$ and exact inverse $\bD^{-1}=\bI$.  The Chebyshev filter uses~$\tilde{\bA}$; the Rayleigh--Ritz step and residual norms use exact~$\bA$.  For $\varepsilon>0$, ChFSI stagnates at a subspace angle of $\mathcal{O}(\varepsilon)$, while R-ChFSI converges to machine precision.\cn}\label{fig:ToyStandard}
\end{figure}

\subsubsection{Generalized eigenproblem with approximate \texorpdfstring{$\bB^{-1}$}{B⁻¹}}\label{sec:toyGen}
We now consider the generalized eigenproblem $\bA\bx=\lambda\bB\bx$ with a non-trivial overlap matrix~$\bB$ ($\kappa_2(\bB)=5$). Because $\bA$ and $\bB$ share the same eigenbasis by construction, the generalized eigenvalues $\lambda_j^{(\mathrm{gen})}=\lambda_j/b_j$ are fully determined by the prescribed spectra of $\bA$ and $\bB$, and the exact generalized eigenpairs are computed independently via LAPACK's \texttt{dsygv} to serve as ground truth. The approximation quality of $\bB^{-1}$ is controlled through $\zeta=\norm{\bD^{-1}-\bB^{-1}}_2$, varying from $0$ (exact inverse) to $10^{-2}$ (an approximation of this magnitude might arise from a coarse lumped-mass approximation).

\Cref{fig:ToyGeneralized} summarises the convergence behavior. Without approximation ($\zeta=0$), both methods converge identically to residuals of order $10^{-13}$. For moderate perturbation levels the two methods diverge markedly. At $\zeta=10^{-4}$ and $\zeta=10^{-3}$, ChFSI stagnates at residuals of $\mathcal{O}(\zeta)$, since the Chebyshev filtering step applies the operator $\bD^{-1}\bA$ repeatedly, amplifying the perturbation $\bD^{-1}-\bB^{-1}$ at each degree of the recurrence. R-ChFSI, on the other hand, reaches residuals of order $10^{-14}$ in both cases, confirming that the residual-based recurrence successfully decouples the converged accuracy from the quality of the approximate inverse.

For the largest perturbation tested ($\zeta=10^{-2}$), the contrast between the two methods remains striking: ChFSI stagnates at a subspace angle of $\mathcal{O}(\zeta)$, whereas R-ChFSI continues to converge monotonically, reaching a subspace angle of approximately $6\times 10^{-8}$ after $100$ iterations. Although the sufficient convergence condition~\cref{eqn:convIneqAppReq} is only marginally satisfied at this noise level, the residual-based recurrence still succeeds in driving the filtering error down in tandem with the residual.  This observation highlights that R-ChFSI substantially extends the range of approximation quality~$\zeta$ for which convergence is achievable, allowing machine-precision results with coarser approximations than ChFSI can tolerate.

\begin{figure}[!ht]
 \includegraphics[width=\textwidth]{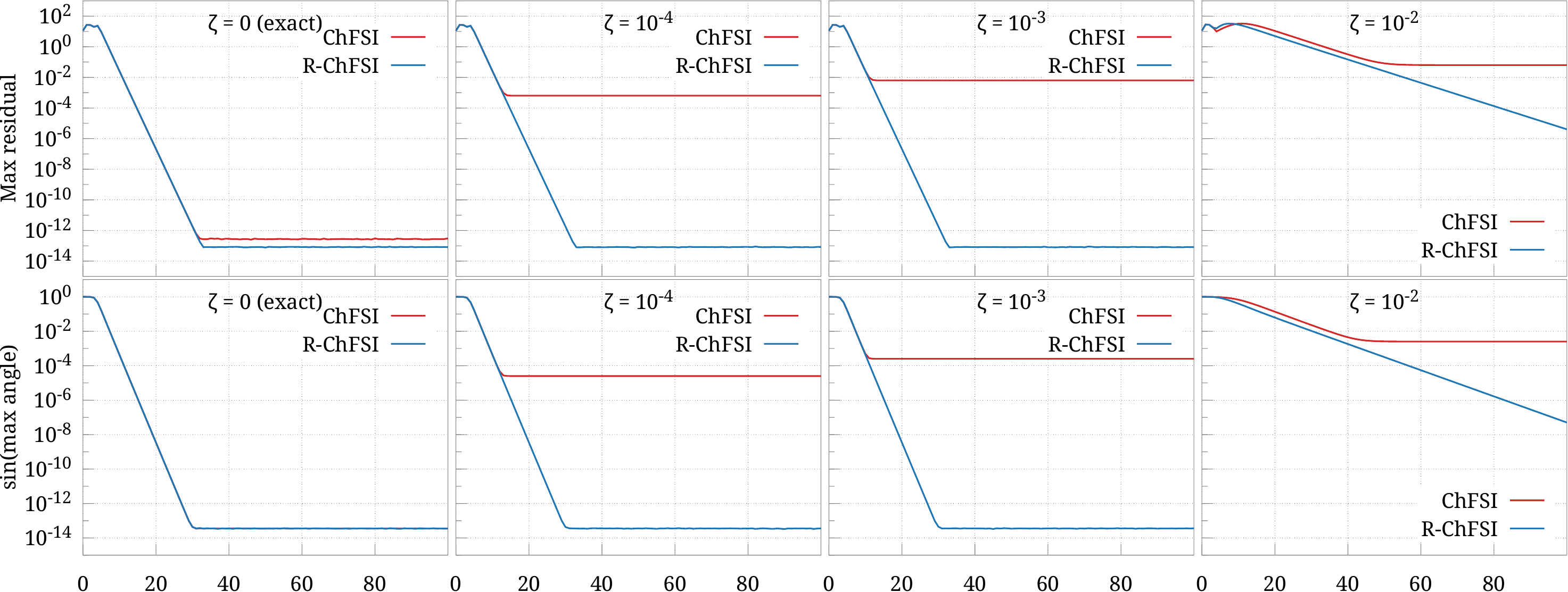}
 \caption{Convergence of ChFSI and R-ChFSI for the generalized eigenproblem ($m=1000$, $n=10$, $p=8$, $\kappa_2(\bB)=5$, spectral gap $\delta=1$) with varying approximation quality $\zeta=\norm{\bD^{-1}-\bB^{-1}}_2\in\{0,10^{-4},10^{-3},10^{-2}\}$.  For $\zeta\leq 10^{-3}$, R-ChFSI converges to machine precision while ChFSI stagnates at a residual of $\mathcal{O}(\zeta)$.  Even at $\zeta=10^{-2}$, R-ChFSI continues to converge (subspace angle $\approx 6\times 10^{-8}$ at iteration $100$) while ChFSI stagnates at $\mathcal{O}(\zeta)$.}\label{fig:ToyGeneralized}
\end{figure}

\subsubsection{Verification of the convergence condition}\label{sec:toyEq5}
The controlled setting of these experiments allows us to verify the convergence condition~\cref{eqn:convIneqAppReq} directly by computing both sides of the inequality at each iteration. We focus on the generalized problem with $\zeta=10^{-3}$, since this noise level is large enough to cause ChFSI to stagnate while still permitting R-ChFSI to converge, making the contrast most informative.

\Cref{fig:ToyEq5} (left panel) plots the filtering error $\norm{\hat{\bDelta}_p^{(i)}}=\norm{\bB^{1/2}\bDelta_p^{(i)}}$ as a function of the iteration index. For ChFSI, this quantity decreases initially (as the subspace rotates towards the eigenspace and the filter becomes more effective) but saturates at $\mathcal{O}(\zeta)$. It cannot decrease further because the filtering step accumulates a noise contribution of magnitude $\mathcal{O}(\zeta)$ at every degree of the Chebyshev recurrence, and this contribution is independent of the current subspace quality. For R-ChFSI, the filtering error tracks the residual norm and decreases monotonically to machine precision, since the recurrence filters residuals rather than eigenvectors. The effective noise injected at each step is therefore proportional to $\norm{\bR^{(i)}}$ and vanishes as convergence progresses.

The right panel of \cref{fig:ToyEq5} displays the left- and right-hand sides of the convergence condition~\cref{eqn:convIneqAppReq}. The left-hand side, $\abs{C_p(\lambda_n)}-\abs{C_p(\lambda_{n+1})}$, is a constant determined solely by the spectrum and the filter parameters. For ChFSI, the right-hand side $\norm{\hat{\bDelta}_p^{(i)}}\!\left(\sec{\angle(\mathcal{S}^{(i)},\mathcal{S})}+\csc{\angle(\mathcal{S}^{(i)},\mathcal{S})}\right)$ initially decreases as the subspace rotates into alignment with the eigenspace but then rises as the subspace angle diminishes further---the trigonometric factor $\sec\theta+\csc\theta$ grows while the filtering error $\norm{\hat{\bDelta}_p^{(i)}}$ stagnates at $\mathcal{O}(\zeta)$. The two curves cross, after which the condition is violated and the subspace angle can no longer decrease---the residual stagnates. For R-ChFSI, the filtering error decreases proportionally to the residual norm, which counteracts the growth of the trigonometric factor, so the right-hand side remains below the left-hand side throughout and convergence proceeds unimpeded to machine precision. These observations confirm, in a quantitative and fully reproducible manner, the central theoretical claim of \cref{sec:apprxSubspaceConstr}: the residual-based reformulation ensures that the filtering error decays in tandem with the residual, thereby preventing the stagnation inherent in the traditional Chebyshev recurrence.

\begin{figure}[!ht]
 \includegraphics[width=\textwidth]{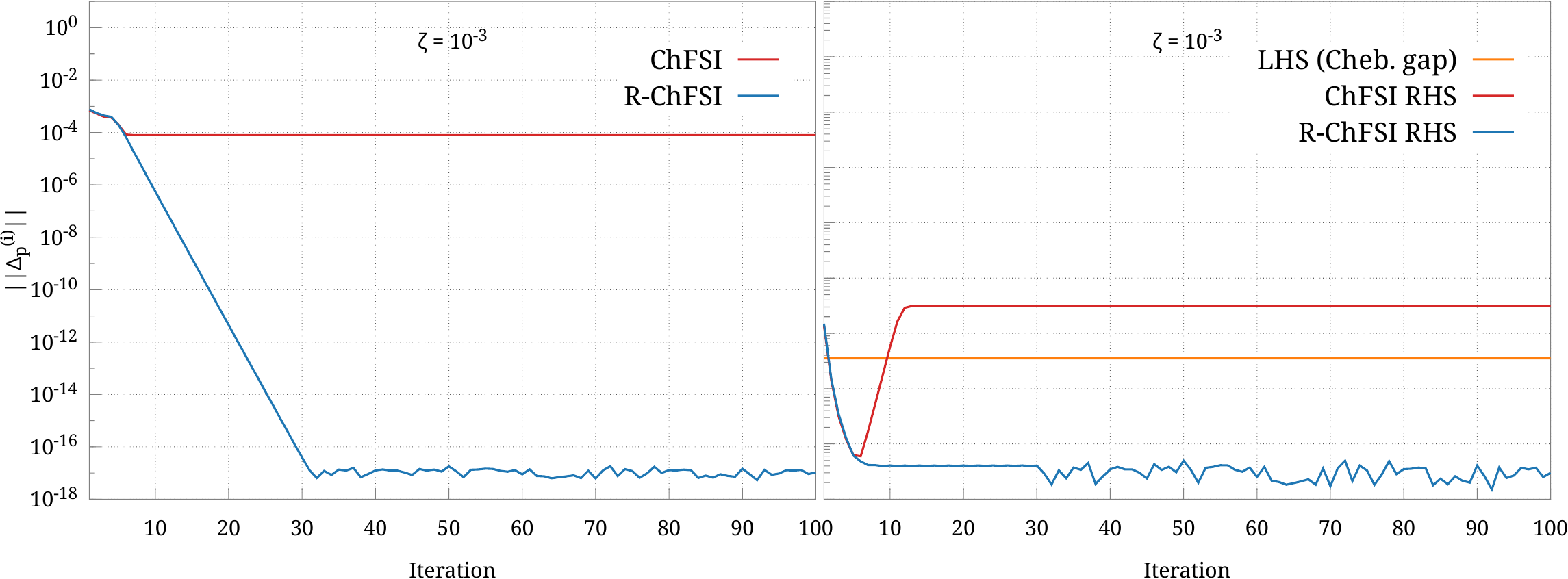}
 \caption{Diagnostic verification of the convergence condition (\cref{eqn:convIneqAppReq}) for the generalized eigenproblem with $\zeta=10^{-3}$, $m=1000$, $n=10$, $p=8$.  Left panel shows the filtering error $\norm{\hat{\bDelta}_p^{(i)}}=\norm{\bB^{1/2}\bDelta_p^{(i)}}$.  For ChFSI it saturates at $\mathcal{O}(\zeta)$, whereas for R-ChFSI it decays monotonically to machine precision.  Right panel plots the left-hand side (constant Chebyshev gap $\abs{C_p(\lambda_n)}-\abs{C_p(\lambda_{n+1})}$) and the right-hand side ($\norm{\hat{\bDelta}_p^{(i)}} (\sec\theta+\csc\theta)$) of~\cref{eqn:convIneqAppReq}.  The condition is eventually violated for ChFSI, whereas it remains satisfied for R-ChFSI throughout.}\label{fig:ToyEq5}
\end{figure}

These controlled experiments establish that the advantages of R-ChFSI over ChFSI arise directly from the algebraic structure of the residual-based recurrence and its robustness to any mismatch between the operator used for filtering and the one used for the Rayleigh--Ritz projection. In the standard eigenproblem, when the Chebyshev filter uses an approximate~$\tilde{\bA}$ while the projection uses exact~$\bA$, ChFSI stagnates at a residual of $\mathcal{O}(\varepsilon)$ whereas R-ChFSI converges to machine precision, since the filtering noise enters the residual-based recurrence proportionally to the residual norm and therefore vanishes during convergence. In the generalized eigenproblem, where $\bD^{-1}\neq\bB^{-1}$, the same mechanism applies: ChFSI stagnates at $\mathcal{O}(\zeta)$ while R-ChFSI converges for all noise levels tested ($\zeta\leq 10^{-2}$), substantially extending the usable range of approximation quality. These observations are fully consistent with the theoretical predictions of \cref{sec:apprxSubspaceConstr}. Having validated the theory on problems with precisely controlled errors, we now proceed to evaluate R-ChFSI on large-scale sparse eigenproblems arising in Kohn-Sham DFT.
\cn}
For all our benchmark studies reported in the remainder of this section, we consider the sparse matrix eigenvalue problems that arise from higher-order finite-element (FE) discretization of Kohn-Sham density functional theory (DFT)~\cite{chenAdaptiveFiniteElement2014,schauerAllelectronKohnSham2013,tsuchidaAdaptiveFiniteelementMethod1996,tsuchidaElectronicstructureCalculationsBased1995,dasDFTFE10Massively2022,motamarriDFTFEMassivelyParallel2020,paskFiniteElementMethods2005,motamarriHigherorderAdaptiveFiniteelement2013}, widely used in \emph{ab initio} modeling of materials. The Kohn-Sham DFT equations, discretized using a non-orthogonal basis, are a canonical example of the large generalized Hermitian nonlinear eigenvalue problem\footnote{The generalized eigenproblem arises in the finite-element discretization of DFT equations because of the non-orthogonality of the FE basis functions.} that motivates polynomial filtered-based methods. The nonlinear eigenvalue problem is solved as a sequence of linear eigenvalue problems using a self-consistent procedure (SCF). Each iteration of SCF requires a solution of the generalized Hermitian sparse eigenproblem for the lowest $n$ eigenpairs with $n<<m$, with $n$ proportional to the number of electrons in the material system and $m$ being the number of FE basis functions. Hence, we believe that Kohn-Sham DFT provides a demanding, realistic testbed for our algorithmic claims. For the sparse generalized eigenvalue problem of the form $\bA \bx = \lambda \bB \bx$ considered here, we explore the use of diagonal approximation of $\bB$ obtained by the mass lumping scheme, commonly employed in finite-element methods, to compute its inverse only during the subspace construction step within both the ChFSI and R-ChFSI methods, showing that the latter can achieve significantly lower residual tolerances. 

%We begin by comparing the R-ChFSI method to the ChFSI method using a diagonal approximation of the finite-element overlap matrix. 

We now evaluate the accuracy and performance of the proposed R-ChFSI method for solving the aforementioned Hermitian definite generalized eigenvalue problems arising in DFT by replacing \cref{alg:chebFilt} in \cref{alg:ChFSI} by \cref{alg:RChFSiGhep}. {\cblue We note that the DFT-FE code already employs several mixed-precision strategies within the eigensolver---for example, mixed-precision in the Rayleigh--Ritz projection and the Gram--Schmidt orthogonalization. Our R-ChFSI modifications are complementary, as they target the \emph{filtering step} (sparse matrix--multi-vector products in the Chebyshev recurrence), which is the dominant cost for the problem sizes of interest and was previously performed entirely in FP64. The mixed-precision strategies already present in DFT-FE therefore remain active and unaffected, while R-ChFSI unlocks additional savings in the computationally most expensive part of the eigensolver.\cn} We implement the proposed R-ChFSI in the open-source code DFT-FE\footnote{https://github.com/dftfeDevelopers/dftfe}, a massively parallel finite-element code written in C++ for first principles based materials modeling using Kohn-Sham density functional theory (DFT). To this end, we consider a set of three benchmark material systems employing periodic boundary conditions, with the sparse matrix dimensions 1.728 million, 2.567 million, 85.766 million seeking 3000, 7000 and 13500 eigenpairs, respectively. These details are summarized in \cref{tab:problemSizes}. The three benchmark systems are chosen to be representative of a wide range of typical finite-element discretized DFT problems, considering the following factors: (a) differing spectral gap between wanted and unwanted {\cblue eigenspectrum\cn}, (b) differing number of target eigenpairs, and (c) both Gamma-point (real symmetric) and k-point (complex Hermitian) matrices. These distinctions matter for filtering-based eigensolvers. Systems with a small or zero spectral gap between wanted and unwanted spectra require higher amplification of the desired eigenspectrum and hence higher polynomial degree to separate near-degenerate states, while systems with a larger gap allow for lower polynomial degree and converge faster. Similarly, targeting a larger number of eigenvectors increases the cost of the Rayleigh-Ritz step (as the Rayleigh-Ritz step scales cubically, compared to the quadratic scaling of the filtering step) and alters the balance of where runtime is spent. Finally, k-point sampling of the Brillouin zone in DFT produces complex Hermitian matrices and different communication/compute patterns compared to Gamma-point runs, demonstrating that both cases ensure the conclusions apply to the typical varieties of DFT calculations encountered in practice.
\begin{table}[htbp]
\cblue\footnotesize
\caption{Dimensions and key spectral properties of the benchmark problems considered. These correspond to 3 material systems comprising (1) $6\times 6 \times 6$ supercell of Molybdenum, (2) $12\times 12 \times 12$ supercell of Silicon and (3) $15\times 15 \times 15$ supercell of Carbon, each with a single vacancy. The spectral gap is defined as $\delta=\lambda_{n+1}-\lambda_n$ where $\lambda_n$ and $\lambda_{n+1}$ are the largest wanted and smallest unwanted eigenvalues, respectively. Subspace fraction denotes $n/m$.}\label{tab:problemSizes}
\begin{center}
 \begin{tabular}{|c|c|c|c|c|c|} \hline
 System & \# of DoFs ($m$) & \# of wanted eigenvectors ($n$) & Subspace dim.\ & Subspace fraction ($n/m$) & Spectral gap ($\delta$)\\ \hline
(1)  & 1728000&3000 &3600 & $1.7\times 10^{-3}$ & $\sim 0$ (metallic)\\
    \hline
(2) & 25672375&7000 &8400 & $2.7\times 10^{-4}$ & $\sim 0.03$ Ha\\
    \hline
 (3) & 85766121&13500 &16800 & $1.6\times 10^{-4}$ & $\sim 0.18$ Ha\\ \hline
 \end{tabular}
\end{center}
\cn\end{table}

We will now demonstrate that the R-ChFSI method converges to a significantly lower residual tolerance than the ChFSI method when the diagonal approximation is employed for approximating the inverse overlap matrix during the subspace construction. To this end we will first consider the ChFSI and R-ChFSI algorithms for the benchmark systems summarized in \cref{tab:problemSizes} using various Chebyshev polynomial filter degrees and compare the residual norm ($\max_j r_j^{(i)}=\max_j\norm{\bA\bx_j^{(i)}-\epsilon_j^{(i)}\bB\bx_j^{(i)}}$) that can be achieved for both the methods. We pick a target residual norm criteria as $\max_j r_j^{(i)}<10^{-8}$, a typical tolerance employed in DFT calculations. We subsequently proceed to analyse the behavior of the R-ChFSI method when low-precision arithmetic is employed during the subspace construction. To this end, we benchmark the robustness and performance of the R-ChFSI algorithm when employing FP32 arithmetic. To this end, in \cref{alg:RChFSiGhep} we store $\bR_{\bX}$ and $\bR_{\bY}$ in FP32 and the computation of $\bA\bD^{-1}\bR_{\bY}$ is done using FP32 arithmetic. Further, we also employ TF32 tensor cores (by using \texttt{oneapi::mkl::blas::compute\_mode::float\_to\_tf32}), and we report the residual norms achieved by the FP32 and TF32 variants of R-ChFSI for our benchmark systems for various values of the Chebyshev polynomial degree. We further observe that using BF16 arithmetic does not provide any noticeable performance improvement during the filtering step and instead slows down the convergence. We therefore perform the computation using TF32 tensor cores and the nearest-neighbour MPI communication required during FE-discretized matrix multi-vector products in BF16 precision; we term this method TF32B for convenience of notation. This further improves the performance of the eigensolver, as will be demonstrated subsequently. For all benchmarking studies, we consider two cases, namely real symmetric and complex Hermitian eigenvalue problems.

\subsection{Real Symmetric Eigenvalue Problems}
We now consider the finite-element discretization of the Kohn-Sham DFT equations sampled at the origin (Gamma-points) of the Brillouin zone for the systems described in \cref{tab:problemSizes}. This ensures that the resulting discretized equation is a real symmetric eigenvalue problem.
\begin{figure}[!ht]
 \includegraphics[width=\textwidth]{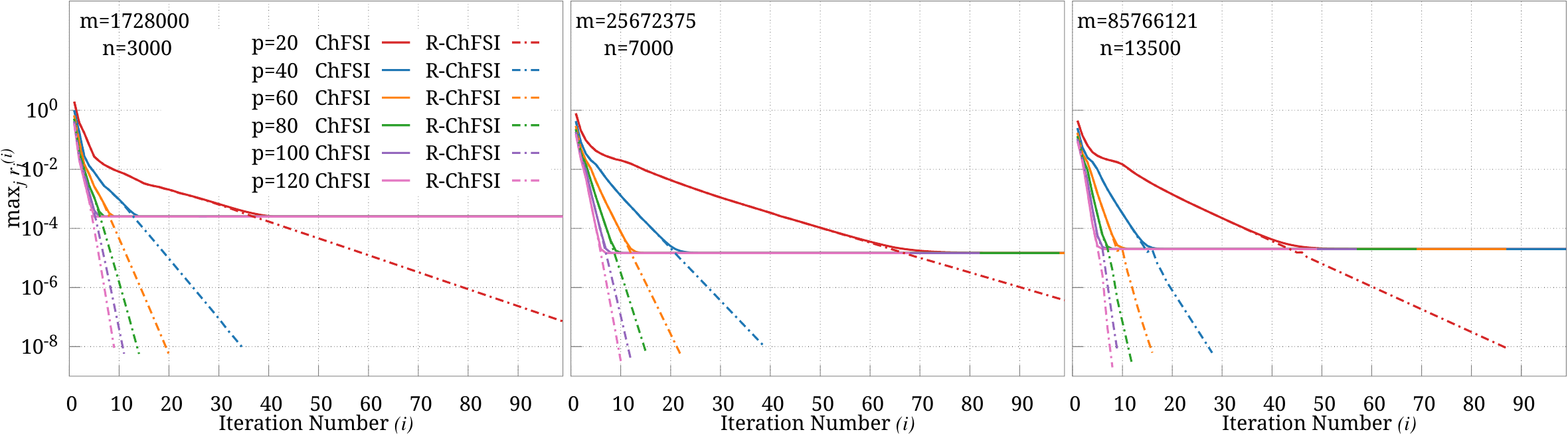}
 \caption{Plot of $\max_j r_j^{(i)}=\max_j\norm{\bA\bx_j^{(i)}-\epsilon_j^{(i)}\bB\bx_j^{(i)}}$ as the iterations progress for the ChFSI method and the R-ChFSI method (both in FP64 arithmetic) with various values of the Chebyshev polynomial filter degree ($p$ in Algorithm 2 and 3) to solve the symmetric generalized eigenvalue problem for the benchmark systems described in \cref{tab:problemSizes}.}\label{fig:ResRealGHEP}
\end{figure}
\begin{figure}[!t]
 \includegraphics[width=\textwidth]{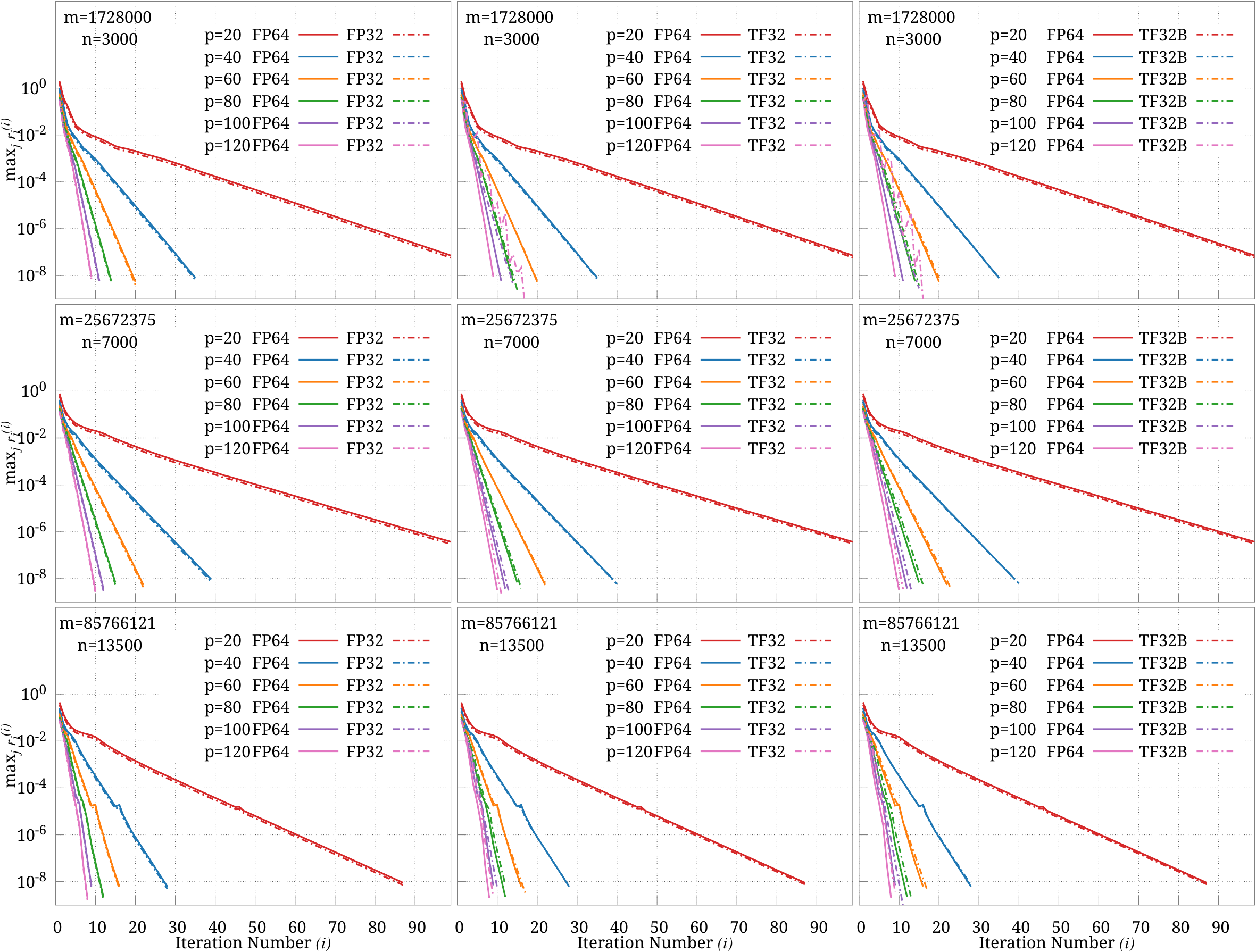}
 \caption{Plot of $\max_j r_j^{(i)}=\max_j\norm{\bA\bx_j^{(i)}-\epsilon_j^{(i)}\bB\bx_j^{(i)}}$ as the iterations progress for the R-ChFSI method to solve the symmetric generalized eigenvalue problem with various precisions for the benchmark systems described in \cref{tab:problemSizes}. A slight offset has been added to the lower precision results for ease of visualization.}\label{fig:ResRealGHEPLP}
 \vspace{0.15in}
 \includegraphics[width=\textwidth]{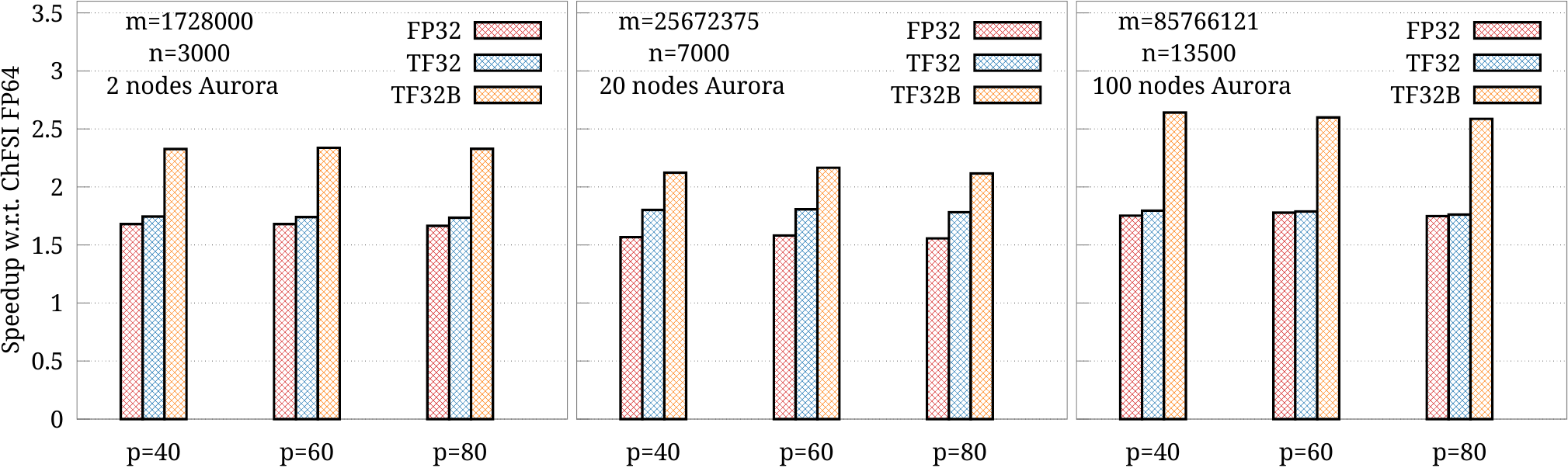}
 \caption{Speedups of lower precision R-ChFSI methods over the FP64 R-ChFSI method for subspace construction to solve the symmetric generalized eigenvalue problem for the benchmark systems described in \cref{tab:problemSizes}.}\label{fig:PerfRealGHEPGPU}
\end{figure}

From \cref{fig:ResRealGHEP}, we observe that for all the benchmark systems summarized in \cref{tab:problemSizes} and various Chebyshev polynomial degrees, the residual tolerance that can be achieved by the R-ChFSI method is orders of magnitude lower than what can be achieved using the ChFSI method when employing FP64 arithmetic and the diagonal approximation for the inverse overlap matrix in the subspace filtering step. This observation is consistent with the analysis done in \cref{sec:apprxSubspaceConstr}. As previously discussed, the R-ChFSI algorithm enables us to construct the filtered subspace using lower precision arithmetic. We report the residual norms ($\max_j r_j^{(i)}=\max_j\norm{\bA\bx_j^{(i)}-\epsilon_j^{(i)}\bB\bx_j^{(i)}}$) achieved by the FP32, TF32 and TF32B variants of R-ChFSI in \cref{fig:ResRealGHEPLP}. We note that the residual norms obtained using the FP32, TF32 and TF32B variants are comparable to those obtained using the FP64 variant for Chebyshev polynomial degrees of $p=20,40,60,80$ but not for $p=100,120$. As demonstrated by our controlled experiments in \cref{sec:toyEq5}, this can be attributed to the convergence condition of \cref{eqn:convIneqAppReq} being violated at higher polynomial degrees, where the accumulated filtering error $\norm{\hat{\bDelta}_p^{(i)}}$ exceeds the left-hand side $\abs{C_p(\lambda_n)}-\abs{C_p(\lambda_{n+1})}$. 

\begin{figure}[!ht]
 \includegraphics[width=\textwidth]{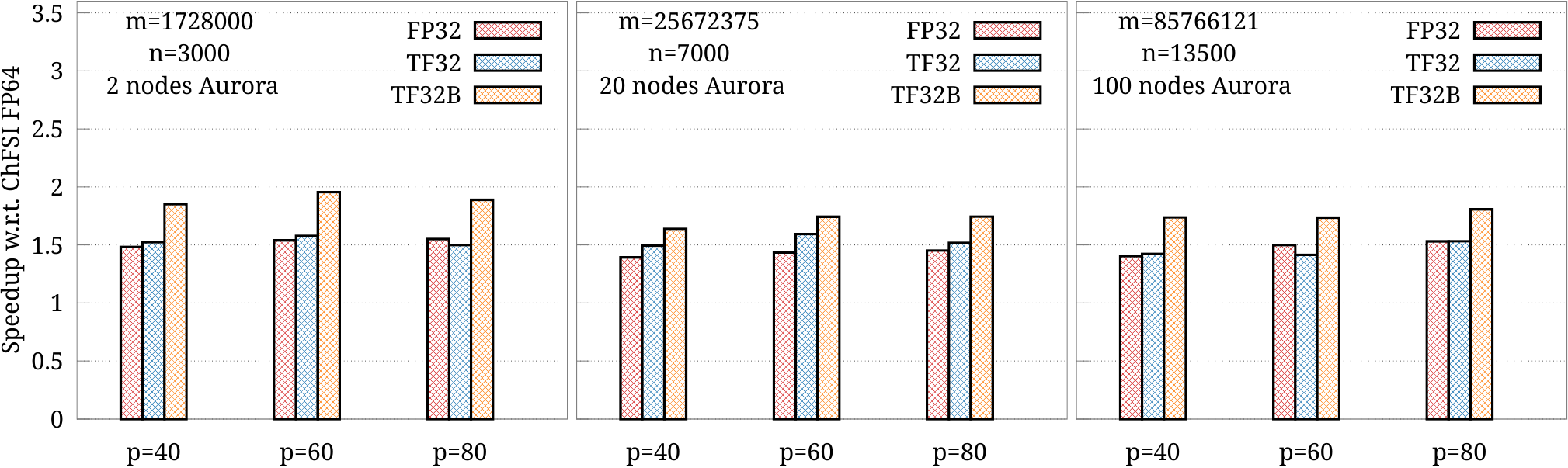}
 \caption{Speedups of lower precision R-ChFSI methods over the FP64 R-ChFSI method for the eigensolver to reach $\max_j r_j^{(i)}=\max_j\norm{\bA\bx_j^{(i)}-\epsilon_j^{(i)}\bB\bx_j^{(i)}}<10^{-8}$ to solve the symmetric generalized eigenvalue problem for the benchmark systems described in \cref{tab:problemSizes}.}\label{fig:PerfRealGHEPTotGPU}
\end{figure}

In \cref{fig:PerfRealGHEPGPU}, we report the speedups achieved for the filtering step for Chebyshev polynomial degrees\footnote{Polynomial degree 20 is not shown in \cref{fig:PerfRealGHEPGPU} as it did not converge to the desired tolerance within the time limit and degrees 100 and 120 are not shown as numerical instabilities start showing up at these polynomial degrees.} of $p=40,60,80$ by the FP32, TF32 and TF32B variants of R-ChFSI over the FP64 variant of R-ChFSI on Intel Data Center GPU Max Series accelerators deployed in the Aurora supercomputing system. We achieve speedups of up to 1.8x for the TF32 variant and 2.6x for the TF32B variant of the R-ChFSI method over the FP64 variant of the R-ChFSI method. We further report in \cref{fig:PerfRealGHEPTotGPU} the speedups achieved for the total eigensolve that includes both filtering and Rayleigh-Ritz step (with the convergence criteria $\max_j r_j^{(i)}=\max_j\norm{\bA\bx_j^{(i)}-\epsilon_j^{(i)}\bB\bx_j^{(i)}}<10^{-8}$) for Chebyshev polynomial degrees of $p=40,60,80$ by the FP32, TF32 and TF32B variants of R-ChFSI over the FP64 variant of R-ChFSI on Intel Data Center GPU Max Series accelerators deployed in the Aurora supercomputing system. For the full eigensolve, we achieve speedups of up to 1.6x for the TF32 and 2.0x for the TF32B variant of the R-ChFSI method over the FP64 variant of the R-ChFSI method.

\begin{remark}[Architecture independence of the speedups]
{\cblue We emphasise that the performance gains reported above are not contingent on any FP64-throughput limitation of the underlying hardware. The computationally dominant operation in the filtering step---the sparse-matrix multi-vector product $\bA\bD^{-1}\bR_{\bY}$---is implemented in the finite-element setting as a sequence of element-level dense matrix--matrix multiplications dispatched via batched GEMM routines (\texttt{stridedBatchedGemm}/\texttt{batchGemm}), which leverage tensor-core hardware (or equivalent matrix engines, e.g.\ Intel XMX units) when operating in TF32. Nevertheless, the operation remains largely memory-bandwidth bound on modern architectures, so the primary performance lever is the reduction in data movement. Operating in lower precision reduces the number of bytes streamed through the memory hierarchy (64~bits per element in FP64 versus 32~bits in FP32/TF32), thereby directly alleviating the memory-bandwidth bottleneck irrespective of the ratio of peak FP64 to FP32 throughput. For TF32, the combination of reduced data movement and higher matrix-engine throughput yields the observed speedups. Consequently, even on nodes with ample double-precision performance the reduced-precision R-ChFSI variants will deliver a speedup that scales roughly with the ratio of data widths. The theoretical upper bound on the filtering-step speedup from a pure data-volume argument is $64/32=2{\times}$ for FP32 and TF32, consistent with the measured speedups of $1.5$--$1.8{\times}$ for the filtering step after accounting for latency and other overheads. The TF32B variant gains additional speedup because the nearest-neighbour MPI communication required during the finite-element matrix--vector product is performed in BF16 (16~bits), reducing inter-node message sizes by a factor of~$4$ relative to FP64.\cn}
\end{remark}

\subsection{Complex Hermitian Eigenvalue Problems}
We now consider the finite-element discretization of the Kohn-Sham DFT equations sampled at a non-zero $k$-point ~\cite{dasDFTFE10Massively2022} in the Brillouin zone for the systems described in \cref{tab:problemSizes}. The resulting discretized equation is now a complex Hermitian eigenvalue problem. 
\begin{figure}[!ht]
 \includegraphics[width=\textwidth]{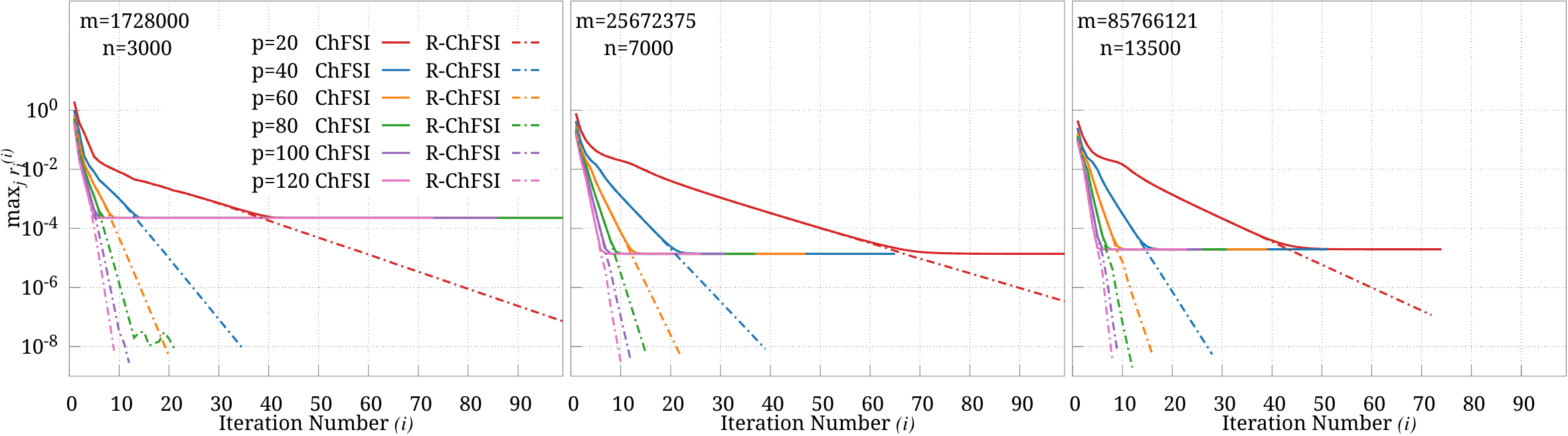}
 \caption{Plot of $\max_j r_j^{(i)}=\max_j\norm{\bA\bx_j^{(i)}-\epsilon_j^{(i)}\bB\bx_j^{(i)}}$ as the iterations progress for the ChFSI method and the R-ChFSI method to solve the Hermitian generalized eigenvalue problem for the benchmark systems described in \cref{tab:problemSizes}.}\label{fig:ResComplexGHEP}
\end{figure}
Even for Hermitian eigenvalue problems, \cref{fig:ResComplexGHEP} shows that across all benchmark systems summarized in \cref{tab:problemSizes} and for the range of Chebyshev polynomial degrees, the R-ChFSI method attains residual tolerances that are orders of magnitude lower than those achievable with the standard ChFSI method when using FP64 arithmetic and the diagonal approximation for the inverse overlap matrix in the filtering step.
\begin{figure}[!ht]
 \includegraphics[width=\textwidth]{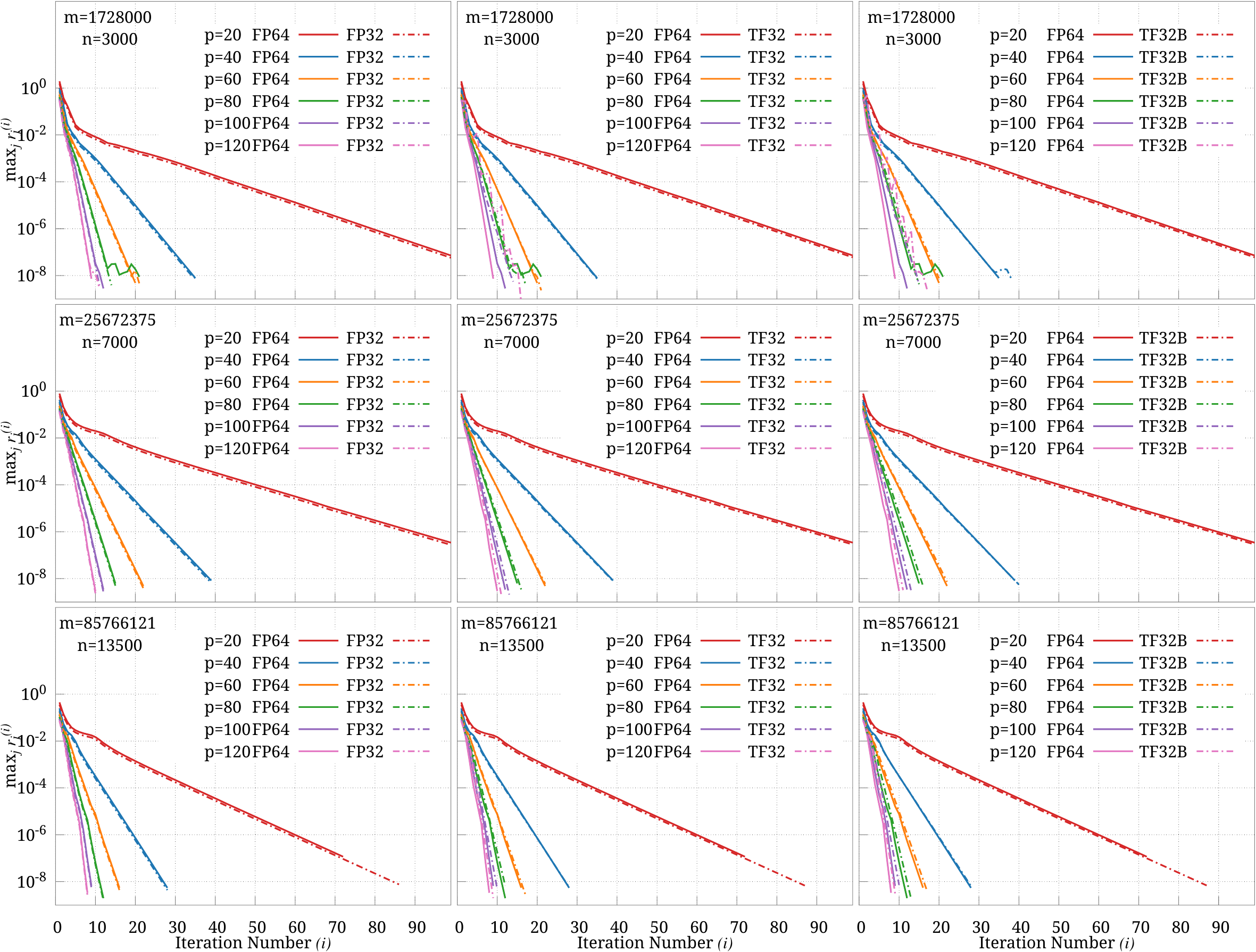}
 \caption{Plot of $\max_j r_j^{(i)}=\max_j\norm{\bA\bx_j^{(i)}-\epsilon_j^{(i)}\bB\bx_j^{(i)}}$ as the iterations progress for the R-ChFSI method to solve the Hermitian generalized eigenvalue problem with various precisions for the benchmark systems described in \cref{tab:problemSizes}. A slight offset has been added to the lower precision results for ease of visualization.}\label{fig:ResComplexGHEPLP}
 \end{figure}
We also note that our proposed R-ChFSI algorithm allows us to construct the filtered subspace using lower precision for the complex Hermitian eigenproblem, and we report the residuals achieved by the FP32, TF32 and TF32B variants of R-ChFSI in \cref{fig:ResComplexGHEPLP}. We note that the residual norms obtained using the FP32, TF32 and TF32B variants are comparable to those obtained using the FP64 variant for Chebyshev polynomial degrees of $p = 20, 40, 60, 80$. 
We also report the speedups achieved for the filtering step for Chebyshev polynomial degrees of $p=40,60,80$ by the FP32, TF32 and TF32B variants of R-ChFSI over the FP64 variant of R-ChFSI in \cref{fig:PerfComplexGHEPGPU} on Intel Data Center GPU Max Series accelerators deployed in the Aurora supercomputing system. We achieve speedups of up to 2.3x for the TF32 variant and 2.7x for the TF32B variant of the R-ChFSI method over the FP64 variant of the R-ChFSI method. We further report the speedups achieved for the total eigensolve (with the convergence criteria $\max_j r_j^{(i)}=\max_j\norm{\bA\bx_j^{(i)}-\epsilon_j^{(i)}\bB\bx_j^{(i)}}<10^{-8}$) for Chebyshev polynomial degrees of $p=40,60,80$ by the FP32, TF32 and TF32B variants of R-ChFSI over the FP64 variant of R-ChFSI in \cref{fig:PerfComplexGHEPTotGPU} on Intel Data Center GPU Max Series accelerators deployed in the Aurora supercomputing system. For the full eigensolve, we achieve speedups of up to 1.9x for the TF32 and 2.1x for the TF32B variant of the R-ChFSI method over the FP64 variant of the R-ChFSI method.
 \begin{figure}[!ht]
 %\vspace{0.15in}
 \includegraphics[width=\textwidth]{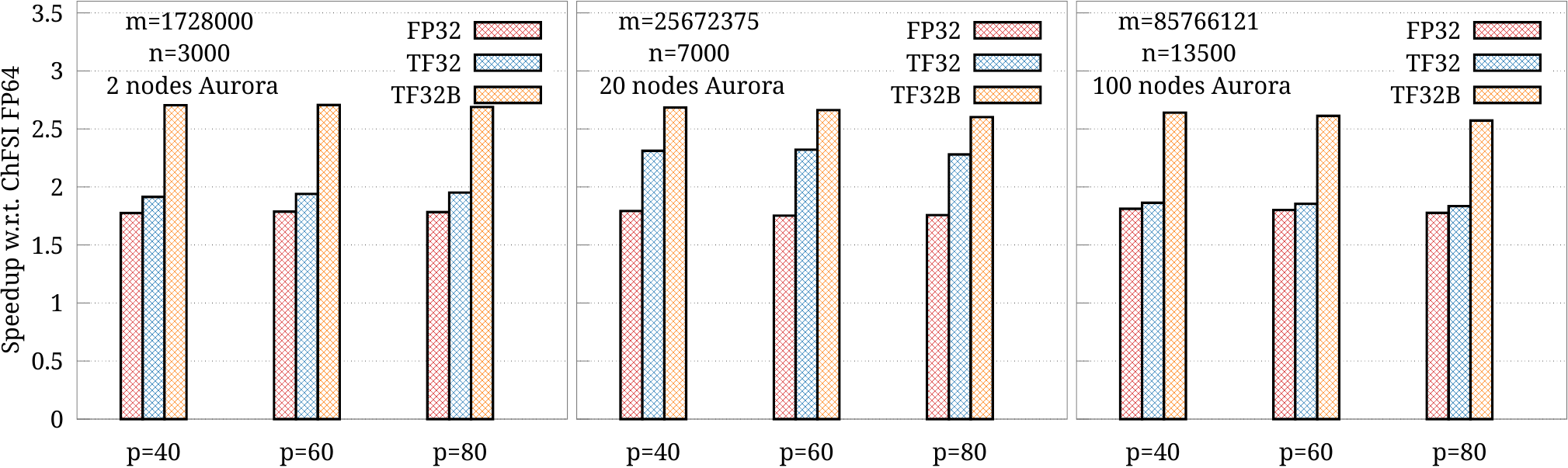}
 \caption{Speedups of lower precision R-ChFSI methods over the FP64 R-ChFSI method for subspace construction to solve the Hermitian generalized eigenvalue problem for the benchmark systems described in \cref{tab:problemSizes}.}\label{fig:PerfComplexGHEPGPU}
\end{figure}
\begin{figure}[!ht]
 \includegraphics[width=\textwidth]{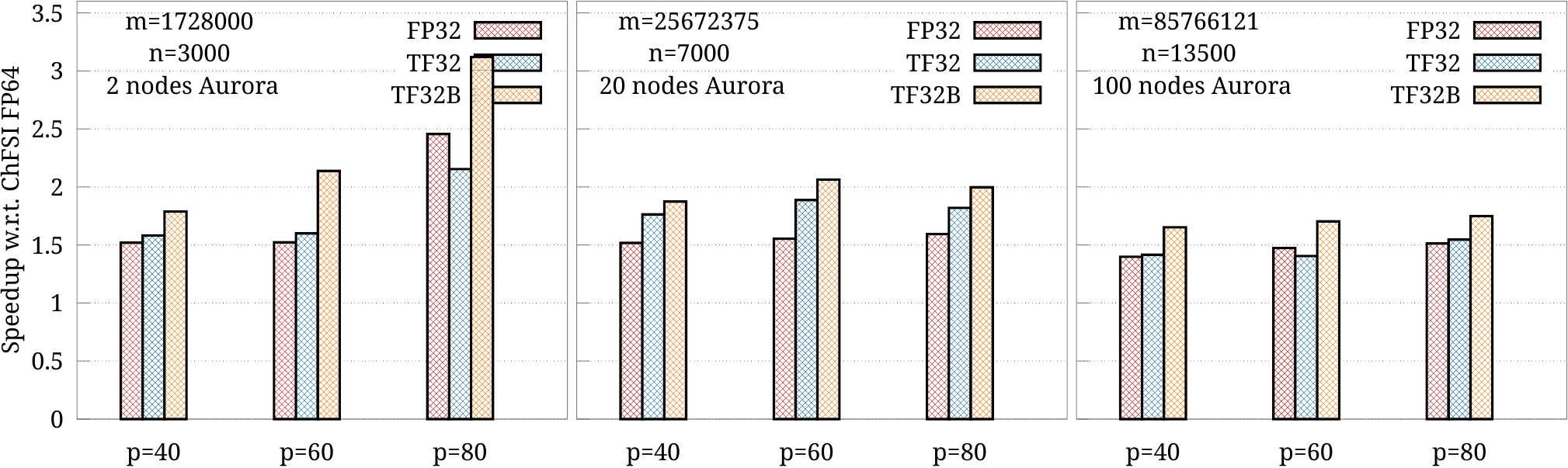}
 \caption{Speedups of lower precision R-ChFSI methods over the FP64 R-ChFSI method for the eigensolver to reach $\max_j r_j^{(i)}=\max_j\norm{\bA\bx_j^{(i)}-\epsilon_j^{(i)}\bB\bx_j^{(i)}}<10^{-8}$ to solve the Hermitian generalized eigenvalue problem for the benchmark systems described in \cref{tab:problemSizes}.}\label{fig:PerfComplexGHEPTotGPU}
\end{figure}

{\cblue \subsubsection{Summary: complex Hermitian vs.\ real symmetric}
Comparing the complex Hermitian results above with the real symmetric results in the preceding subsection, three observations stand out. First, the \emph{accuracy} benefit of R-ChFSI over ChFSI is equally pronounced in both cases: the residual-based reformulation reduces the stagnation floor by several orders of magnitude regardless of whether the matrices are real or complex, confirming that the theoretical guarantees of \cref{sec:apprxSubspaceConstr} hold without modification for the complex Hermitian setting. Second, the lower-precision variants (FP32, TF32, TF32B) remain as robust for complex matrices as for real ones, converging to the same residual tolerances at polynomial degrees $p\leq80$. Third, the \emph{performance} speedups are uniformly larger in the complex case (up to $2.3{\times}$ for TF32 and $2.7{\times}$ for TF32B in the filtering step, versus $1.8{\times}$ and $2.6{\times}$ in the real case). This is expected: complex arithmetic doubles the number of floating-point operations and the volume of data per matrix element, so halving the precision yields a proportionally greater reduction in both memory traffic and computation. The full-eigensolve speedups ($1.9{\times}$/$2.1{\times}$ for complex vs.\ $1.6{\times}$/$2.0{\times}$ for real) follow the same trend, with the slightly smaller gains reflecting the fixed FP64 cost of the Rayleigh--Ritz projection. Overall, these results demonstrate that R-ChFSI with low-precision arithmetic is effective across the full range of DFT calculations---both $\Gamma$-point (real symmetric) and $k$-point (complex Hermitian).\cn}

\section{Conclusion}
In this work, we have proposed a residual-based Chebyshev filtered subspace iteration (R-ChFSI) method designed for large-scale Hermitian eigenvalue problems, with a particular emphasis on robustness to inexact matrix–vector products. Through both mathematical analysis and extensive numerical experiments, we have demonstrated that the proposed method outperforms the standard Chebyshev filtered subspace iteration (ChFSI) approach significantly. Our theoretical justification demonstrates how reformulating the Chebyshev recurrence to operate on residuals, rather than directly on eigenvector estimates, effectively suppresses numerical errors introduced by inexact operator applications. This leads to a more reliable construction of the filtered subspace when compared with the traditional ChFSI recurrence relation.

An important contribution of this work lies in demonstrating the effectiveness of R-ChFSI for generalized Hermitian eigenvalue problems of the form $\bA \bx = \lambda \bB \bx$. In this setting, the proposed method is particularly impactful due to its ability to incorporate approximate inverses of $\bB$ when constructing the subspace rich in the desired eigenspace using the Chebyshev polynomial recurrence relation. By replacing the exact computation of matrix inverses with computationally inexpensive approximate inverses (e.g., diagonal or block-diagonal approximations of $\bB$), the method not only reduces the overall computational burden but also maintains convergence robustness. This results in substantial performance gains, as the proposed approach achieves residual tolerances that are orders of magnitude lower than those obtainable with traditional ChFSI when solving the generalized problems. This highlights the method’s robustness and its suitability for large-scale discretized generalized eigenproblems where exact factorizations of \(\bB\) are prohibitively expensive.

Another important aspect of the proposed approach is its natural compatibility with low-precision arithmetic. By design, R-ChFSI maintains stable convergence even when matrix–vector products are evaluated using lower-precision formats such as FP32 or TF32. Our results demonstrate that the method retains residual tolerances comparable to those obtained using full double precision, while benefiting from the reduced computational cost and increased throughput associated with low-precision operations. This makes R-ChFSI particularly attractive for emerging high-performance computing platforms, including GPU-accelerated systems and specialized tensor-processing hardware, which increasingly prioritize mixed-precision arithmetic to deliver high performance. The ability to jointly leverage approximate inverses and low-precision arithmetic allows R-ChFSI to achieve significant computational speedups while preserving numerical robustness.

The proposed R-ChFSI methodology has the potential to significantly impact a wide range of computational physics applications. These include \emph{ab initio} modeling of materials, structural mechanics, fluid dynamics, and multiphysics problems, all of which rely on the efficient solution of large Hermitian eigenvalue problems arising from discretizations of partial differential equations using finite-element, finite-difference, wavelets or finite-volume formulations. By enabling the effective use of low-precision arithmetic without sacrificing accuracy, the method provides a compelling alternative for large-scale simulations that require substantial computational resources. Moreover, its robustness in solving large generalized nonlinear eigenvalue problems opens new opportunities for applying R-ChFSI to domains such as quantum modeling of materials, and other areas that demand scalable and reliable eigensolvers.

\section{Acknowledgements}
The authors gratefully acknowledge the seed grant from the Indian Institute of Science and the SERB Startup Research Grant from the Department of Science and Technology, Government of India (Grant Number: SRG/2020/002194). This research used resources of the Argonne Leadership Computing Facility (ALCF), a US Department of Energy (DOE) Office of Science user facility at Argonne National Laboratory (ANL) and is based on research supported by the U.S DOE Office of Science-Advanced Scientific Computing Research Program, under Contract No. DE-AC02-06CH11357. PM acknowledges Vishwas Rao from ANL for useful discussions on the manuscript. The research also used the resources of PARAM Pravega at the Indian Institute of Science, supported by the National Supercomputing Mission (NSM), for testing a few of the benchmarks on CPU clusters. N.K. and K.R. would like to acknowledge the Prime Minister Research Fellowship (PMRF) from the Ministry of Education, India, for financial support. P.M. acknowledges the Google India Research Award 2023 for financial support during the course of this work.

\section{Declaration of Generative AI and AI-assisted technologies in the writing process}
During the preparation of this work the authors used Paperpal in order to proofread and Claude (Anthropic, via github copilot) in order to assist with coding for dense eigenproblem testing. After using these tools/services, the authors reviewed and edited the content as needed and take full responsibility for the content of the publication.

\section{Data availability}
The code and data associated with this article can be found online at \url{https://github.com/matrixlabiisc/R-ChFSI-Manuscript-Data}.
\newpage
\appendix
\vspace{0.1in} 
\begin{center}
\textbf{Appendix}
\end{center}
\section{Useful Lemmas}
\begin{lemma}
 If $\normalfont\bH\in\mathbb{C}^{m\times m}$ and $\normalfont\bX\in\mathbb{C}^{m\times n}$, the spectral norm of the error in evaluating the matrix product $\normalfont\bH\bX$ due to floating point approximations satisfies $\normalfont\norm{\bH\bX-{\bH\otimes\bX}}\leq\gamma_m\norm{\bH}\norm{\bX}$\label{lem:FPErrorMatmul}
\end{lemma}
\begin{proof}
 This is a straightforward application of the results derived in \cite[section 3.5]{highamAccuracyStabilityNumerical2002}. We note that $\gamma_m$ is a constant that depends on $m$, and the machine precision of the floating point arithmetic (denoted as $\varepsilon_M$). For the case of dense real matrices, we have $\gamma_m=m^{2}\varepsilon_M/(1-m\varepsilon_M)$ and for dense complex matrices, it needs to be appropriately modified as described in \cite[section 3.6]{highamAccuracyStabilityNumerical2002}. We also note that tighter bounds can be obtained from probabilistic error analysis \cite{ipsenProbabilisticErrorAnalysis2020,highamSharperProbabilisticBackward2020}. Further, we also note that in the case where $\bH$ is sparse, significantly tighter bounds independent of $m$ can be achieved depending on the specific implementations of the sparse matrix-dense matrix multiplication routines. 
\end{proof}
\begin{lemma}
 Consider a recurrence relation of the form
\begin{align}
\normalfont\bDelta_{k+1}^{(i)}=a_k\bH\bDelta_{k}^{(i)}+b_k\bDelta_{k}^{(i)}+c_k\bDelta_{k-1}^{(i)}+a_k\be_k^{(i)}\label{eqn:errRecGen}
\end{align}
with $\normalfont\norm{\be_k^{(i)}}\leq h_0\norm{\bDelta_{k}^{(i)}}+h_1^{(i)}\;\forall k=1,2,\dots,p$. We can then write
  \begin{align}
  \norm{\bDelta_{k+1}^{(i)}}\leq g_k^0\left(\norm{\bDelta_0^{(i)}}+\norm{\bDelta_1^{(i)}}+g_k^1h_1^{(i)}\right)
 \end{align}\label{thm:errGen}
\end{lemma}
\begin{proof}
From \cref{eqn:errRecGen} we can write 
\begin{align*}
 \norm{\bDelta_{k+1}^{(i)}}&\leq \abs{a_k}\norm{\bH}\norm{\bDelta_{k}^{(i)}}+\abs{b_k}\norm{\bDelta_{k}^{(i)}}+\abs{c_k}\norm{\bDelta_{k-1}^{(i)}}+\abs{a_k}\norm{\be_k^{(i)}}\\
 &\leq(\abs{a_k}\norm{\bH}+\abs{b_k}+h_0)\norm{\bDelta_{k}^{(i)}}+\abs{c_k}\norm{\bDelta_{k-1}^{(i)}}+\abs{a_k}h_1^{(i)}
\end{align*}
 Defining the following 
 \begin{align*}
  \bd_k^{(i)}&=\begin{bmatrix}
   \norm{\bDelta_{k}^{(i)}}\\\norm{\bDelta_{k-1}^{(i)}}
  \end{bmatrix} & \bF_k&=\begin{bmatrix}
   \abs{a_k}\norm{\bH}+\abs{b_k}+\abs{a_k}h_0 & \abs{c_k}\\1 & 0
  \end{bmatrix} & \bE_k^{(i)}&=\begin{bmatrix}
   \abs{a_k}h_1^{(i)}\\0
  \end{bmatrix}
 \end{align*}
 we can now write
 \begin{align*}
  \norm{\bd_{k+1}^{(i)}}&\leq\norm{\bF_k\bd_{k}^{(i)}+\bE_k^{(i)}}\leq\norm{\bF_k}\norm{\bd_{k}^{(i)}}+\abs{a_k}h_1^{(i)}
 \end{align*}
 Consequently, we now have
 \begin{align*}
  \norm{\bd_{k+1}^{(i)}}&\leq\left(\prod_{j=1}^k\norm{\bF_j}\right)\norm{\bd_1^{(i)}}+\sum_{j=1}^k\left(\prod_{r=j+1}^k\norm{\bF_r}\right)\abs{a_j}h_1^{(i)}
 \end{align*}
 We note that $\norm{\bDelta_{k+1}^{(i)}}\leq\norm{\bd_{k+1}^{(i)}}$ and $\norm{\bd_1^{(i)}}\leq\norm{\bDelta_0^{(i)}}+\norm{\bDelta_1^{(i)}}$. This results in
 \begin{align*}
  \norm{\bDelta_{k+1}^{(i)}}\leq\left(\prod_{j=1}^k\norm{\bF_j}\right)\left(\norm{\bDelta_0^{(i)}}+\norm{\bDelta_1^{(i)}}\right)+\sum_{j=1}^k\left(\prod_{r=j+1}^k\norm{\bF_r}\right)\abs{a_j}h_1^{(i)}
 \end{align*}
 Defining $g_k^0=\prod_{j=1}^k\norm{\bF_j}$ we have
  \begin{align*}
  \norm{\bDelta_{k+1}^{(i)}}&\leq g_k^0\left(\norm{\bDelta_0^{(i)}}+\norm{\bDelta_1^{(i)}}\right)+\sum_{j=1}^k\frac{g_k^0}{g_j^0}\abs{a_j}h_1^{(i)}\\
  &=g_k^0\left(\norm{\bDelta_0^{(i)}}+\norm{\bDelta_1^{(i)}}+\sum_{j=1}^k\frac{\abs{a_j}}{g_j^0}h_1^{(i)}\right)
 \end{align*}
 Defining $g_k^1=\sum_{j=1}^k\frac{\abs{a_j}}{g_j^0}$ we have
  \begin{align*}
  \norm{\bDelta_{k+1}^{(i)}}&\leq g_k^0\left(\norm{\bDelta_0^{(i)}}+\norm{\bDelta_1^{(i)}}+g_k^1h_1^{(i)}\right)
 \end{align*}
\end{proof}
 \begin{lemma}
 The spectral norm of the residual, $\normalfont\bR_k^{(i)}=C_k(\bH)\bX^{(i)}-\bX^{(i)} C_k(\bLam^{(i)})$ for $k=0,\dots,p$, is bounded by $\normalfont\norm{\bR_k^{(i)}}\leq f_k\norm{\bR^{(i)}}$ where $f_k$ is a finite constant and $\normalfont\bR^{(i)}=\bH\bX^{(i)}-\bX^{(i)}\bLam^{(i)}$\label{lem:rkbound}
\end{lemma}
\begin{proof}
 Let $C_k(x)=\sum_{j=0}^k\alpha_jx^j$, we can now write
 \begin{align*}
  \bR_k^{(i)}&=\sum_{j=1}^k\alpha_j\left(\bH^j{\bX^{(i)}}-{\bX^{(i)}}{\bLam^{(i)}}^j\right)=\sum_{j=1}^k\alpha_j\sum_{r=0}^{j-1}\bH^{j-r-1}\left(\bH{\bX^{(i)}}-{\bX^{(i)}}{\bLam^{(i)}}\right){\bLam^{(i)}}^r
 \end{align*}
 Since $\bLam^{(i)}$ comprises Ritz values of the Hermitian matrix $\bH$, each diagonal entry of $\bLam^{(i)}$ lies in the interval $[\lambda_1,\lambda_m]$ and hence $\norm{\bLam^{(i)}}=\max_j|\epsilon_j^{(i)}|\leq \max\{|\lambda_1|,|\lambda_m|\}\leq\norm{\bH}$.
 Using submultiplicativity of the spectral norm we can now write
 \begin{align*}
  \norm{\bR_k^{(i)}}&\leq\sum_{j=1}^k\sum_{r=0}^{j-1}|\alpha_j|\norm{\bH}^{j-r-1}\norm{\bH{\bX^{(i)}}-{\bX^{(i)}}{\bLam^{(i)}}}\norm{\bLam^{(i)}}^r\leq\sum_{j=1}^k\sum_{r=0}^{j-1}|\alpha_j|\norm{\bH}^{j-1}\norm{\bR^{(i)}}=f_k\norm{\bR^{(i)}}
 \end{align*}
 where we have defined $f_k=\sum_{j=1}^k j|\alpha_j|\norm{\bH}^{j-1}$.
\end{proof}
\begin{lemma}
 The norm of the residual 
 $\normalfont\bR^{(i)}$ satisfies $  \normalfont\norm{\bR^{(i)}}\leq 2\norm{\bH}\sin{\angle(\mathcal{S}^{(i)},\mathcal{S})}$\label{lem:TwoNormBound}
\end{lemma}
\begin{proof}
 We have, $\bR^{(i)}=\bH\bX^{(i)}-\bX^{(i)}\bLam^{(i)}$ and the subspace diagonalization performed in the Rayleigh-Ritz projection step associated with $(i-1)^{th}$ iteration allows one to write $\bLam^{(i)}$ as $\bLam^{(i)}={\bX^{(i)}}^\dagger\bH\bX^{(i)}$. Subsequently we can write the residual $\bR^{(i)}=(\bI-\bX^{(i)}{\bX^{(i)}}^\dagger)\bH\bX^{(i)}$. We note that $\bH=\hat{\bU}_1\bLam_1\hat{\bU}_1^\dagger+\hat{\bU}_2\bLam_2\hat{\bU}_2^\dagger$ and consequently we can write
 \begin{multline*}
  \norm{\bR^{(i)}}=\norm{(\bI-\bX^{(i)}{\bX^{(i)}}^\dagger)\bH\bX^{(i)}}=\norm{(\bI-\bX^{(i)}{\bX^{(i)}}^\dagger)(\hat{\bU}_1\bLam_1\hat{\bU}_1^\dagger+\hat{\bU}_2\bLam_2\hat{\bU}_2^\dagger)\bX^{(i)}}\\
  \leq\norm{(\bI-\bX^{(i)}{\bX^{(i)}}^\dagger)\hat{\bU}_1}\norm{\bLam_1}\norm{\hat{\bU}_1^\dagger\bX^{(i)}}+\norm{(\bI-\bX^{(i)}{\bX^{(i)}}^\dagger)\hat{\bU}_2}\norm{\bLam_2}\norm{\hat{\bU}_2^\dagger\bX^{(i)}}
 \end{multline*}
 Using the fact that $\sin{\angle(\mathcal{S}^{(i)},\mathcal{S})}={\norm{(\bI-\bX^{(i)}{\bX^{(i)}}^\dagger)\hat{\bU}_1} = \norm{\hat{\bU}_2^\dagger\bX^{(i)}}}$ and the inequalities $\norm{\hat{\bU}_1^\dagger\bX^{(i)}}\leq1$, $\norm{(\bI-\bX^{(i)}{\bX^{(i)}}^\dagger)\hat{\bU}_2}\leq1$ we can write
 \begin{align}
  \norm{\bR^{(i)}}&\leq\sin{\angle(\mathcal{S}^{(i)},\mathcal{S})}(\norm{\bLam_1}+\norm{\bLam_2})
\leq2\norm{\bH}\sin{\angle(\mathcal{S}^{(i)},\mathcal{S})}
 \end{align}
\end{proof}

\section{Proof of \texorpdfstring{\Cref{thm:deltaNChFSI}}{Theorem 3.2}}\label{prf:deltaNChFSI}
Using the exact recurrence relation \cref{eqn:ChFSIrecurrence} and the recurrence relation employing inexact matrix products \cref{eqn:ChFSIrecurrenceSEPMP} we can now write a recurrence relation for $\bDelta_k^{(i)}=\underline{\bY}_k^{(i)}-\bY_k^{(i)}$ as
\begin{align}
\bDelta_{k+1}^{(i)}=a_k\bH\bDelta_{k}^{(i)}+b_k\bDelta_{k}^{(i)}+c_k\bDelta_{k-1}^{(i)}+a_k\be_k^{(i)} \label{eqn:DeltaChFSI}
\end{align}
with the initial conditions $\bDelta_0^{(i)}=0$ and $\bDelta_1^{(i)}=\frac{\sigma_1}{e}(\bD^{-1}-\bB^{-1})\bA\bX^{(i)}$, obtained from the initial conditions of \cref{eqn:ChFSIrecurrence} and \cref{eqn:ChFSIrecurrenceSEPMP}. While $\be_k^{(i)}=\bD^{-1}\otimes\bA\otimes\underline{\bY}_k^{(i)}-\bH\underline{\bY}_k^{(i)}$.
From \cref{lem:FPErrorMatmul} we have $\norm{\be_k^{(i)}}=\norm{{\bD^{-1}\otimes\bA\otimes\underline{\bY}_k^{(i)}}-\bH\underline{\bY}_k^{(i)}}\leq(2\gamma_m\norm{\bD}^{-1}+\zeta)\norm{\bA}\norm{\underline{\bY}_k^{(i)}}$. We further note that 
\begin{align*}
 \norm{\underline{\bY}_k^{(i)}}=\norm{\bDelta_k^{(i)}+\bY_k^{(i)}}\leq\norm{\bDelta_k^{(i)}}+\norm{\bY_k^{(i)}}\leq\norm{\bDelta_k^{(i)}}+\norm{C_k(\bH)\bX^{(i)}}\leq\norm{\bDelta_k^{(i)}}+\norm{C_k(\bH)} 
\end{align*}
Thus we have 
$\norm{\be_k^{(i)}}\leq(2\gamma_m\norm{\bD}^{-1}+\zeta)\norm{\bA}\norm{\bDelta_k^{(i)}}+(2\gamma_m\norm{\bD}^{-1}+\zeta)\norm{\bA}\norm{C_k(\bH)}$. Further, since $\hat{\bX}^{(i)}=\bB^{\frac{1}{2}}\bX^{(i)}$ is orthonormal, we have $\norm{\bX^{(i)}}\leq\norm{\bB^{-\frac{1}{2}}}$. Using \cref{thm:errGen} with $\norm{\bDelta_0^{(i)}}=0$ and $\norm{\bDelta_1^{(i)}}\leq\abs{\frac{\sigma_1}{e}}\zeta\norm{\bA}\norm{\bB^{-\frac{1}{2}}}$, $h_0=(2\gamma_m\norm{\bD}^{-1}+\zeta)\norm{\bA}$ and $h_1^{(i)}=(2\gamma_m\norm{\bD}^{-1}+\zeta)\norm{\bA}\norm{C_k(\bH)}$ we obtain $\norm{\hat\bDelta_k^{(i)}}\leq\gamma_m\eta_k+\zeta\tilde{\eta}_k$ with $\eta_k=2\norm{\bB^{\frac{1}{2}}}\norm{\bD^{-1}}\norm{\bA}\norm{C_k(\bH)}g_{k-1}^1$ and $\tilde{\eta}_k=\norm{\bB^{\frac{1}{2}}}(\norm{\bA}\norm{C_k(\bH)}g_{k-1}^1+g_{k-1}^0\abs{\frac{\sigma_1}{e}}\norm{\bA}\norm{\bB^{-\frac{1}{2}}})$.
\section{Proof of \texorpdfstring{\Cref{thm:dkboundGen}}{Theorem 3.4}}\label{prf:dkboundGen}
Using the exact reformulated recurrence relation \cref{eqn:rChFSIrecurrence} and the reformulated recurrence relation employing inexact matrix products \cref{eqn:rChFSIrecurrenceMP} we can now write a recurrence relation for the error as $\bDelta_k^{(i)}=\underline{\bY}_k^{(i)}-\bY_k^{(i)}=\bD^{-1}(\underline{\bZ}_k^{(i)}-\bZ_k^{(i)})$ as
\begin{align*}
\bDelta_{k+1}^{(i)}=a_k\bH\bDelta_{k}^{(i)}+b_k\bDelta_{k}^{(i)}+c_k\bDelta_{k-1}^{(i)}\!+\!a_k\be_k^{(i)}
\end{align*}
with $\be_k^{(i)}=\left(\bD^{-1}-\bB^{-1}\right)\bB\bR^{(i)}\bLam_k^{(i)}+\bD^{-1}\bA\otimes\bD^{-1}\otimes\underline{\bZ}_k^{(i)}-\bB^{-1}\bA\bD^{-1}\underline{\bZ}_k^{(i)}$ and the initial conditions $\bDelta_0^{(i)}=0$ and $\bDelta_1^{(i)}=\frac{\sigma_1}{e}(\bD^{-1}-\bB^{-1})\bB\bR^{(i)}$.
We note that this recurrence relation is of the same form as that of \cref{eqn:errRecGen} in \cref{thm:errGen} with
\begin{align*}
 \norm{\be_k^{(i)}}\leq\zeta(\norm{\bB}\norm{\bR^{(i)}}\norm{\bLam_k^{(i)}}+\norm{\bA}\norm{\bD^{-1}\underline{\bZ}_k^{(i)}})+2\gamma_m\norm{\bA}\norm{\bD^{-1}}^2\norm{\bD}\norm{\bD^{-1}\underline{\bZ}_k^{(i)}}
\end{align*}
 Further, since $\bLam_k^{(i)}=C_k(\bLam^{(i)})$ and $\norm{\bLam^{(i)}}\leq\norm{\bH}$, we have $\norm{\bLam_k^{(i)}}\leq \max_{|x|\leq\norm{\bH}}|C_k(x)|=:M_k$.
 We further note that using \cref{lem:rkbound} we can write
 \begin{equation*}
 \norm{\bD^{-1}\underline{\bZ}_k^{(i)}}=\norm{\bDelta_k^{(i)}+\bD^{-1}\bZ_k^{(i)}}\leq\norm{\bDelta_k^{(i)}}+f_k\norm{\bD^{-1}\bB}\norm{\bR^{(i)}}
 \end{equation*}
 Thus we have 
 \begin{multline*}
 \norm{\be_k^{(i)}}\leq(\zeta+2\gamma_m\norm{\bD^{-1}}^2\norm{\bD})\norm{\bA}\norm{\bDelta_k^{(i)}}\\+((\zeta+2\gamma_m\norm{\bD^{-1}}^2\norm{\bD}) f_k\norm{\bA}\norm{\bD^{-1}\bB}+\zeta\norm{\bB}M_k)\norm{\bR^{(i)}}
 \end{multline*}
 Now, using \cref{thm:errGen} with $\norm{\bDelta_0^{(i)}}=0$ and $\norm{\bDelta_1^{(i)}}\!\!=\!\norm{\frac{\sigma_1}{e}(\bD^{-1}-\bB^{-1})\bB\bR^{(i)}}\!\!\leq\abs{\frac{\sigma_1}{e}}\zeta\norm{\bB}\norm{\bR^{(i)}}$, $h_0=(\zeta+2\gamma_m\norm{\bD^{-1}}^2\norm{\bD})\norm{\bA}$ and $h_1^{(i)}=((\zeta+2\gamma_m\norm{\bD^{-1}}^2\norm{\bD}) f_k\norm{\bA}\norm{\bD^{-1}\bB}+\zeta\norm{\bB}M_k)\norm{\bR^{(i)}}$ we obtain $\norm{\hat\bDelta_k^{(i)}}\leq(\gamma_m\eta_k+\zeta\tilde{\eta}_k)\norm{\bR^{(i)}}$ with $\eta_k=2f_k\norm{\bB^{\frac{1}{2}}}\norm{\bD^{-1}}^2\norm{\bD}\norm{\bA}\norm{\bD^{-1}\bB}g_{k-1}^1g_{k-1}^0$ and $\tilde{\eta}_k=\norm{\bB^{\frac{1}{2}}}(\abs{\frac{\sigma_1}{e}}\norm{\bB}g_{k-1}^0+(f_k\norm{\bA}\norm{\bD^{-1}\bB}+\norm{\bB}M_k)g_{k-1}^1g_{k-1}^0)
 $.
\newpage
\bibliographystyle{elsarticle-num-names} 
\bibliography{references.bib}
\end{document}